%% file: pr26.tex
\documentclass[10pt,a4paper]{article}
\usepackage{graphicx}

\title{On the conservation of specific energy and entropy \\ in infinite anharmonic systems}

\date{\today}

\usepackage[margin=3cm]{geometry}

\usepackage[T1]{fontenc}
\usepackage{microtype}
\usepackage{amsmath,amsfonts,amssymb,amsthm,mathtools,bbm}
\usepackage{braket}

\usepackage{pgfplots}
\pgfplotsset{compat=1.18}

\usepackage[shortlabels]{enumitem}
\usepackage{nameref}
\makeatletter
\def\namedlabel#1#2{\begingroup
   \def\@currentlabel{#2}%
   \label{#1}\endgroup
}
\makeatother

\usepackage{booktabs}
\usepackage[font=small,margin=0.1\linewidth]{caption}

\usepackage{graphicx,xcolor}
\usepackage[colorlinks=true,hyperindex=true]{hyperref}
\usepackage{nameref}

\definecolor{newdarkblue}{RGB}{24,64,184}
\definecolor{newgreen}{RGB}{80,170,55}
\definecolor{newdarkbluec}{RGB}{0,0,55}
\definecolor{newgreenc}{RGB}{0,55,0}
\hypersetup{
    colorlinks,%
    citecolor=newdarkblue,%
    filecolor=red,%
    linkcolor=newdarkblue,%
    urlcolor=newgreen,%
    pdfnewwindow=true,%
    pdfstartview={FitH}
}

\numberwithin{equation}{section}

\newtheorem{theorem}{Theorem}[section]
    \newtheorem{proposition}[theorem]{Proposition}
    \newtheorem{lemma}[theorem]{Lemma}
    \newtheorem{corollary}[theorem]{Corollary}

\theoremstyle{definition}
    \newtheorem{definition}[theorem]{Definition}
    \newtheorem{remark}[theorem]{Remark}
    \newtheorem{example}[theorem]{Example}

\def\rr{{\mathbf R}}
\def\zz{{\mathbf Z}}
\def\nn{{\mathbf N}}

\newcommand{\bq}{\mathbf{q}}
\newcommand{\bp}{\mathbf{p}}

	\makeatletter
	\providecommand*{\diff}%
	{\@ifnextchar^{\DIfF}{\DIfF^{}}}
	\def\DIfF^#1{%
	\mathop{\mathrm{\mathstrut d}}%
	\nolimits^{#1}\gobblespace}
	\def\gobblespace{%
	\futurelet\diffarg\opspace}
	\def\opspace{%
	\let\DiffSpace\!%
	\ifx\diffarg(%
	\let\DiffSpace\relax
	\else
	\ifx\diffarg%
	\let\DiffSpace\relax
	\else
	\ifx\diffarg\{%
	\let\DiffSpace\relax
	\fi\fi\fi\DiffSpace}
\newcommand{\dd}{\diff}

\newcommand{\Exp}[1]{\mathrm{e}^{#1}}
\renewcommand{\complement}{\mathsf{c}}
\renewcommand{\Set}{\Delta}
\newcommand{\irange}{D}
\newcommand{\amu}{\bar \mu}

\DeclareMathOperator{\vol}{vol}

\newcommand{\ellip}{\textnormal{[...]}}

\newcommand{\nien}{E_0^{\textnormal{n.i.}}}
\newcommand{\nieni}{E_i^{\textnormal{n.i.}}}
\newcommand{\nienj}{E_j^{\textnormal{n.i.}}}

\newcommand{\nifn}{F_0^{\textnormal{n.i.}}}

\newcommand{\res}[1]{\upharpoonright_{#1}}

\makeatletter
\def\namedlabel#1#2{\begingroup
   \def\@currentlabel{#2}%
   \label{#1}\endgroup
}
\makeatother

\begin{document}

%\maketitle
\author{Gaia Pozzoli\textsuperscript{a} and Renaud Raqu\'epas\textsuperscript{b}}
\date{}

\maketitle

\begin{center}
  \begin{minipage}[b]{0.45\textwidth}
    \small
    \centering
    a. Universit\`a degli Studi di Milano-Bicocca\\
    Dipartimento di Matematica e Applicazioni\\
    via R. Cozzi 55, 20125 Milano, Italy
  \end{minipage} 
  \hfill
  \begin{minipage}[b]{0.45\textwidth}
    \small
    \centering
    b. Duke University \\
    Department of Mathematics \\
    Durham NC, United States
  \end{minipage}
\end{center}

\begin{abstract}
    We work with infinite, closed, translation-invariant, finite-range lattice systems with ``unbounded classical spins'', also known as anharmonic crystals, under assumptions close to those used by Lanford, Lebowitz and Lieb (\emph{J.\ Stat.\ Phys.}, 1977)% to lay the foundations of existence and uniqueness for time evolution
    ; among other conditions, the pinning dominates the interaction. In this context, we prove conservation of the specific energy and specific entropy under the time evolution, and we discuss their relation to approach to thermal equilibrium, paralleling known results in the theory of quantum spin systems, where noncommutativity, as opposed to lack of compactness, is the main source of difficulties. 

    \medskip


    %37K99 % None of the above, but in "Dynamical system aspects of infinite-dimensional Hamiltonian and Lagrangian systems"
\end{abstract}

\section{Introduction}

The results of this note are motivated by a desire to better understand the role of conservation laws in what is sometimes called the Gibbs postulate of approach to equilibrium~\cite{DS78}. Although this work fits into a program initiated by our colleagues in~\cite{JPT24a,JPSTp} 
related to some questions posed by Ruelle~\cite{R67}, we summarize the context for the reader's convenience.

In the context of closed, translation-invariant lattice systems\,---\,which is the setup with which we will be working\,---, this postulate can be formulated as follows: 
Barring the presence of additional integrals of motion, one expects that, as far as the statistics of local observables are concerned, the time evolution of a reasonable translation-invariant initial state rapidly approaches a translation-invariant \emph{thermal equilibrium} state compatible with the conservation of energy.
Rigorously proving a statement of this nature is, even in toy models, a notoriously hard problem in mathematical physics. 
Here, ``thermal equilibrium'' can be understood in terms of the energy-constrained entropy maximization problem\footnote{Both entropy and energy should be understood per unit volume.}, or equivalently in terms of the Gibbs variational principle at a certain inverse temperature, leaving aside for the moment the relation to Bowen--Gibbs states, weak Gibbs states in the sense of Yuri and Gibbs states in the sense of the Dobrushin--Lanford--Ruelle equations. In a slightly different context, this characterization in terms of entropy\,---\,which is always conserved in finite-dimensional Hamiltonian systems by Liouville's theorem\,---\,led Ruelle to ask whether the evolution of a closed infinite system should increase its entropy per unit volume, or keep constant entropy per unit volume at every finite time, but converge to a limit with strictly larger entropy per unit volume~\cite[\S{13.iv}]{R67}. We recall that such a jump of entropy is possible because entropies per unit volume tend to only be upper semicontinuous. This echoes early comments from Gibbs himself on the ``index of probability'', whose ``average'' is minus the entropy: 
\begin{quote}
    {\em

    Let us next consider whether an ensemble of isolated  systems has any tendency in the course of time toward a  state of statistical equilibrium. 

    There are certain functions of phase which are constant in time. The distribution of the ensemble with respect to the values of these functions is necessarily invariable \ellip. The distribution in phase which without violating this condition gives the least value of the average index of probability of phase is \ellip{}
    that in which the index of probability is a function of the functions mentioned. \ellip
    
    It would seem, therefore, that we might find a sort of measure of the deviation of an ensemble from statistical equilibrium in the excess of the average index above the minimum which is consistent with the condition of the invariability of the distribution with respect to the constant functions of phase. But we have seen that the index of probability is constant in time for each system of the ensemble. The average index is therefore constant, and we find by this method no approach toward statistical equilibrium in the course of time. 
    
    Yet we must here exercise great caution. One function may approach indefinitely near to another function, while some quantity determined by the first does not approach the corresponding quantity determined by the second.
    }
    \hfill\cite[Ch.\,XII]{Gib}
\end{quote}

In the context of quantum spin systems, it was shown one year after Ruelle's article that, under reasonable assumptions, entropy per unit volume is conserved by the finite-time evolution~\cite{LR68}; also see~\cite{LR72}. One of our central results is that such a result also holds in the context of translation-invariant classical lattice systems, even though different difficulties arise due to the unbounded one-body phase space. Our other central result concerns the conservation of energy per unit volume that is implicitly part of the above formulation of the Gibbs postulate.

Such results will serve as a building block for a research program on the approach to equilibrium in translation-invariant, classical systems, paralleling that initiated in~\cite{JPT24a,JPSTp} 
for quantum spin systems. These works provide structural results relating the property of approach to equilibrium to instantaneous and short-time dynamical properties (thermodynamic formalism, regularity, entropy balance, conservation laws, etc.).
More generally, a motivation for our work is to renew interest in mathematical questions that, within the context of classical, closed, translation-invariant lattice Hamiltonian systems, combine dynamical and thermodynamic considerations; see e.g.~\cite{MPPS79}. A likely explanation for the relative scarcity of results in this direction is that pillars of the literature on classical statistical mechanics such as~\cite{Rue,Isr,Sim,vEFS93,P02} are predominantly concerned with discrete-spin or bounded-spin models, where questions of Hamiltonian dynamics are often less natural; notable exceptions include~\cite{R76,LP76,P80,K81,Geo}.

\paragraph*{Organization of the article} We settle various preliminary questions in Section~\ref{sec:prelim}: space of configurations, assumptions on the interactions, elements of thermodynamic formalism (with references to Appendix~\ref{app:ent} for technicalities related to the failure of the reference volume measure on the one-particle phase space to be normalizable), and existence and uniqueness of the time evolution when starting with configurations whose local energy does not grow too fast across the lattice (a variant of results from~\cite{LLL77,MPPS79}). Sections~\ref{sec:energy} and~\ref{sec:entropy} contain our results on conservation of specific quantities under the time evolution: the former establishes conservation of specific energy; the latter, conservation of the specific entropy. {A discussion of equilibrium in the sense of the variational principle, as well as of the problem of the approach to thermal equilibrium, is given in Section~\ref{sec:DES}.} 

\paragraph*{Acknowledgements} GP was partially funded by the CY Initiative of Excellence through the grant Investissements d'Avenir ANR-16-IDEX-0008 while GP was a post-doctoral researcher at CY Cergy Paris Universit\'e. This research is part of GP's activity within the Group ``DinAmicI'' (www.dinamici.org). RR was partially funded by the Fonds de recherche du Qu\'ebec (section Nature et technologies). Both authors would like to thank V.~Jak\v{s}i\'{c} and C.-A.~Pillet for instrumental discussions on the content of their work~\cite{JPT24a,JPSTp}. RR would like to thank A.C.D.~van Enter for insightful exchanges regarding the literature on Gibbsianity  and systems with unbounded spins.

\section{Preliminaries}
\label{sec:prelim}

\subsection{State space and interactions}
\label{ssec:spaces}

We consider the lattice $\zz^\nu$ with $\nu \in \nn$ and the state space
\begin{align*}
    \Omega := \bigtimes_{i \in \zz^\nu} T^*X_i
\end{align*}
where each $X_i$ is a copy of a common Euclidean space or torus (the configuration space of an individual oscillator or rotator, respectively) and $T^*X_i$ is its cotangent bundle. The set $\Omega$ is equipped with the product topology~$\mathcal{T}$ and product $\sigma$-algebra~$\mathcal{F}$.\footnote{Because $T^*X_i$ is a separable metric space and $\zz^\nu$ is countable, it is a basic fact that the Borel $\sigma$-algebra~$\mathcal{F}$ for the product topology coincides with the product $\sigma$-algebra built from the Borel $\sigma$-algebras on each space.} We use $(\bq,\bp) = (q_i,p_i)_{i\in\zz^\nu}$ for the canonical coordinates on~$\Omega$. 
The lattice $\zz^\nu$ acts on itself by translation: $(i,j) \mapsto \mathsf{T}_i(j) := j+i$. We then use $\mathsf{T}^*_i : \Omega \to \Omega$  for the map defined by $(\mathsf{T}^*_i(\bq,\bp))_j = (q_{j+i},p_{j+i})$. When the coordinate $p_j$ or $q_j$ is viewed as a function on $\Omega$, the relevant translation relation is $p_j \circ \mathsf{T}^*_i = p_{j-i}$ or $q_j \circ \mathsf{T}^*_i = q_{j-i}$.

A \emph{family of interactions} is an assignment of a measurable function 
\(
    W_{\Set} : \bigtimes_{i\in {\Set}} X_i \to \rr
\)
to each finite subset~${\Set} \Subset \zz^\nu$. Each such function can be extended in a natural way to a measurable function on~$\Omega$. 
Intuitively, we want to think of the family of interactions as providing a formal Hamiltonian 
\[ 
    H = \sum_{i \in \zz^\nu} K_{\{i\}} + \sum_{\Set \Subset \zz^\nu} W_\Set,
\]
where $K_{\{i\}} : (\bq,\bp) \mapsto \tfrac 12 |p_i|^2$ is the kinetic energy of the particle at site $i$.

Based on a family of interactions, we introduce a convenient way to keep track of distances between elements of the underlying lattice~$\zz^\nu$.

\begin{definition}[Interaction distances]
    First, $\gamma(i,i) = 0$ and $\gamma(i,j) = 1$ if $i \neq j$ and there exists $\{i,j\} \subseteq {\Set} \Subset \zz^\nu$ such that $W_{\Set} \not\equiv 0$. Then, inductively, $\gamma(i,j) = k+1$ if $\gamma(i,j) > k$ and there exists $j'$ such that $\gamma(i,j') = k$ and $\gamma(j',j) = 1$. This extends as usual to a notion of distance between a point and a set:
    $
        \gamma(i,{\Set}) = \inf_{j\in {\Set}} \gamma(i,j).
    $
    We then set
    $
        \gamma({\Set}',{\Set}) = \inf_{i \in {\Set}'} \gamma(i,{\Set}).
    $%
    \footnote{This is \emph{not} the Hausdorff distance between sets.}
\end{definition}

\begin{remark}
    In the case of nearest-neighbour pair interactions, $\gamma(i,j) = |i-j|_{\ell^1(\zz^\nu)}$.
\end{remark}

We now start introducing assumptions that will play an important role throughout. The first ones are very common and easy to interpret.

\begin{description}
    \item[{F}\namedlabel{it:F}{F}] The family of interactions has \emph{finite range} in the sense that there exists $\irange \in \nn$ such $\gamma(i,j) > 1$ whenever $|i-j| > \irange$.
    
    \item[{R}\namedlabel{it:R}{R}] The interactions are \emph{regular} in the sense that the function~$W_{\Set}$ is twice continuously differentiable for all~${\Set} \Subset \zz^\nu$.
    
    \item[{TI}\namedlabel{it:TI}{TI}] The interactions are \emph{translation invariant} in the sense that
    $W_{\mathsf{T}_i(\Set)} = W_{\Set} \circ \mathsf{T}^*_{-i}$ 
    for every $i$ and every $\Set \Subset \zz^\nu$. 
\end{description}

The other assumptions to which we are going to appeal are more technical, and we suggest, on a first reading, focusing on the compact case discussed in Example~\ref{rem:simplify-torus} below or the polynomial case in Example~\ref{ex:polynomial}.

\begin{description}
    \item[{P1}\namedlabel{it:P1}{P1}] The on-site potential $W_{\{i\}}$ is nonnegative for every $i$.
    
    \item[{P2}\namedlabel{it:P3}{P2}] 
    The on-site potentials are pinning in the sense that there exists constants $a$ and $b$ with the property that the sublevel sets $L_i(z) = \{q_i \in X_i : {W_{\{i\}}}(q_i) \leq z\}$ are compact and satisfy
    $$\operatorname{conv}\left(L_i(z)\right) \subseteq L_i(az+b) \Subset X_i$$ 
    for all $i$ and all $z \geq 0$.\footnote{If $X_i$ is a Euclidean space, then ``$\operatorname{conv}$'' takes the meaning of the usual convex hull of a set. If $X_i$ is a torus, then we can declare that ``$\operatorname{conv}$'' of any nonempty set is the full torus, which makes~(\ref{it:P3})  trivial contingent on~(\ref{it:R}); see Remark~\ref{rem:simplify-torus}.}

    \item[{P3}\namedlabel{it:P4}{P3}]  The on-site potentials are pinning in the sense that there exists a constant $C_0$ such that
    \[
    |q_i|\leq C_0 W_{\{i\}}(\bq)
    \]
    for every $i$ and every $\bq$.
    
    \item[{D1}\namedlabel{it:D1}{D1}] The interaction forces are strongly dominated by the pinning potentials in the sense that there exists a constant $C_1$ such that
    \begin{align*}
        \left|\sum_{\substack{\Set \subseteq \Lambda \\ {\Set} \supsetneq \{i\} }}\partial_{q_i} W_{\Set}(\bq)\right|^2
        \leq  C_1 \sum_{\substack{j \in \Lambda \\ \gamma(i,j) \leq 1}} W_{\{j\}}(\bq).
    \end{align*}
    for every $i$ and every $\Lambda \Subset \zz^\nu$.

    \item[{D2}\namedlabel{it:D2}{D2}] The interaction energies are dominated by the pinning potentials in the sense that there exists a constant $C_2$ such that
    \begin{align*}
        \left|\sum_{\substack{\Set \subseteq \Lambda \\ {\Set} \supsetneq \{i\} }}W_{\Set}(\bq)\right|
        \leq  C_2 \sum_{\substack{j \in \Lambda \\ \gamma(i,j) \leq 1}} W_{\{j\}}(\bq).
    \end{align*}
    for every $i$ and every $\Lambda \Subset \zz^\nu$.

    \item[{D3}\namedlabel{it:D3}{D3}] The interaction energies are dominated by the pinning potentials in the sense that there exists a constant $C_3$ such that
    \begin{align*}
        \max_{j \in \Lambda} \left|\sum_{\substack{\Set \subseteq \Lambda \\ {\Set} \supseteq \{i\}}}\partial_{q_j}\partial_{q_i} W_{\Set}(\bq)\right|
        &\leq C_3 \sum_{\substack{j \in \Lambda \\ \gamma(i,j) \leq 1}} W_{\{j\}}(\bq)
    \end{align*}
    for every $i$ and every $\Lambda \Subset \zz^\nu$.
\end{description}

\begin{remark}
\label{rem:add-cst}
    If of any help for Assumptions~(\ref{it:P1})--(\ref{it:P4}) or (\ref{it:D1})--(\ref{it:D3}), one can always add a constant to each one-body potential~$W_{\{i\}}$, without any physical or dynamical consequence.
\end{remark}

\begin{example}
\label{rem:simplify-torus}
    If the configuration spaces $X_i$ are copies of a torus of any dimension\,---\,sometimes referred to as the rotator case\,---, then~(\ref{it:F}),~(\ref{it:R}) and~(\ref{it:TI}) imply~(\ref{it:P1})--(\ref{it:P4}) and~(\ref{it:D1})--(\ref{it:D3}) up to an inconsequential additive constant to $W_{\{0\}}$ as in Remark~\ref{rem:add-cst}.
\end{example}

\begin{example}
\label{ex:polynomial}
    As in~\cite{LLL77}, a familiar example to keep in mind is the case of $X_i = \rr$, a one-body potential that is a polynomial of degree $2n$ with a positive leading coefficient, and finite-range, two-body interaction potentials that are multinomials of degrees at most $n+1$ in the difference. This of course includes the harmonic case ($n=1$), but it is important that we are able to consider anharmonic cases as well. Adding smooth, compactly supported perturbations to these polynomials respecting~(\ref{it:TI}) does not spoil any of the other assumptions, up to an inconsequential additive constant to $W_{\{0\}}$ as in Remark~\ref{rem:add-cst}.
\end{example}

We introduce the local energy-per-particle functions at the site~$i$: 
\begin{align*}
    \nieni &:= K_{\{i\}} + W_{\{i\}}, 
    & 
    E_i &:= K_{\{i\}} + \sum_{\Set \ni i} \frac{1}{\#\Set} W_{\Set}.
\end{align*}
The superscript ``$^{\textnormal{n.i.}}$'' stands for ``noninteracting''. 
We will use the same notation for these objects whether they are seen as functions on $\Omega$, or on some $\bigtimes_{j\in\Lambda} T^*X_j$ with $\Lambda$ large enough that $\Lambda \supseteq \Set$ for all $\Set \ni i$ with nontrivial $W_{\Set}$. 
Under  (\ref{it:TI}), we have $\nieni = \nien \circ \mathsf{T}_{-i}^*$ and $E_i = E_0 \circ \mathsf{T}_{-i}^*$.
Also, formally,  
\begin{align*}
    H = \sum_{i \in \zz^\nu} E_i.
\end{align*}

We let $(\mathfrak{A}_{r},\|\,\cdot\,\|_r)$ be a normed space of arrays indexed by~$\zz^\nu$ with $\mathfrak{A}_{r}$ the space of all~$\boldsymbol{\xi} : \zz^\nu \to \rr$ for which 
\[ 
    \|\boldsymbol{\xi}\|_r := \sup_{i \in \zz^\nu} \frac{|\xi_i|}{(1+|i|)^r}
\]
is finite.
Define 
\begin{align}\label{eq:B-en}
    \mathfrak{B}_{r,C} &:= \left\{(\bq,\bp) : \|\nieni(\bq,\bp)\|_r \leq C \right\}, 
    &
    \mathfrak{B}_{r} &:= \bigcup_{C > 0} \mathfrak{B}_{r,C} =\left\{(\bq,\bp) : \|\nieni(\bq,\bp)\|_r <\infty \right\}.
\end{align}
and
\begin{align*}
    \Omega_2 &:= \bigcup_{r \in (0,2)} \mathfrak{B}_{r}.
\end{align*}
Note that, for every $r>0$, the set $\mathfrak{B}_{r}$ is invariant under $\mathsf{T}_i^*$; so is $\Omega_2$.

The set 
$\Omega_2$ will serve as a set of configurations on which a dynamical group $(\tau_t)_{t\in\rr}$  can be unambiguously defined; this is the content of Lemma~\ref{lem:dyn-bounds}. We have not been aiming for the largest possible such set of configurations (cf.~\cite{LLL77}), but we are satisfied with Definition~\ref{def:I2} and {Corollary}~\ref{cor:important-I2} below, which in particular say the following: under our assumptions, in dimensions $\nu=1,2,3$, almost all sample configurations from a translation invariant state giving finite mean and variance to $\nien$ lie in~$\Omega_2$.
{Definition~\ref{def:I2} shares similarities with the complement of ``bad configurations'' of~\cite{LP76}, which are pointwise sets of large local energy assigned small weight by Gibbs measures, and with the notion of tempered configurations in the Gibbs literature~\cite{R70,LP76,Geo}, which require controlled growth at infinity. 

We define so-called $a$-\emph{elementary boxes}, indexed by a scale length $a\in\nn$:
\begin{align*}
    \Lambda(a) := (\hspace{-.6ex}( -a, a]\!]^\nu.
\end{align*}
This choice is made so that $\#\Lambda(a) = 2^\nu a^{\nu}$.
We also define the (inner) \emph{interaction boundary}: 
\begin{align*}
    \partial^{\textnormal{int}}\Lambda(a) 
    := \{i \in \Lambda(a) : \gamma(i,\Lambda(a)^\complement) = 1\}.
\end{align*}
This is the collection of sites in the box $\Lambda(a)$ that directly interact with the outside of that box. Under the finite-range assumption~(\ref{it:F}), we have $\#\partial^{\textnormal{int}}\Lambda(a) = O(a^{\nu-1}) = o(\# \Lambda(a))$ as $a\to\infty$.

\subsection{Elements of thermodynamic formalism}

An \emph{extended state} on~$\Omega$ is a probability measure on~$(\Omega,\mathcal{F})$. An extended state~$\mu$ is called \emph{translation invariant}, written $\mu \in \mathcal{I}$, if $\mu \circ (\mathsf{T}^*_i)^{-1} = \mu$ for all $i \in \zz^\nu$. 
An extended state~$\mu$ induces a consistent family~$(\mu\res{\Lambda})_{\Lambda\Subset\zz^\nu}$ of marginals: the marginal $\mu\res{\Lambda}$ is a probability measure on $\bigtimes_{i\in\Lambda} T^*X_i$.
For a fixed finite subset~$\Lambda$ of the lattice~$\zz^\nu$, the \emph{finite-volume entropy} is
\begin{align*}
    S_\Lambda(\mu\res{\Lambda}) :=
        \begin{cases}
            \int -\ln \frac{\dd \mu\res{\Lambda}}{\dd \vol_\Lambda} \dd\mu\res{\Lambda}
                & \text{ if } \mu\res{\Lambda} \ll \vol_\Lambda \text{ and } (-\ln \frac{\dd \mu\res{\Lambda}}{\dd \vol_\Lambda})_+ \in L^1(\dd\vol_\Lambda), \\
            \infty 
                & \text{ otherwise.}
        \end{cases}
\end{align*}
Because the volume measure on $\bigtimes_{i\in\Lambda} T^*X_i$\,---\,denoted $\vol_\Lambda$ and acting as a reference measure\,---\,is not normalizable, the quantity $S_\Lambda(\mu\res{\Lambda})$ can be negative. We deal in Appendix~\ref{app:ent} with the consequences of this for the entropy per unit volume
$$
    s(\mu) = \lim_{a\to\infty} \frac{S_{\Lambda(a)}(\mu\res{\Lambda(a)}) }{\#\Lambda(a)},
$$
including existence of the limit. 

\begin{lemma}
\label{lem:ien-from-nien}
    Suppose (\ref{it:F}), (\ref{it:P1}), (\ref{it:D2}) and (\ref{it:TI}) hold and let $\zeta > 0$. If $\mu$ is translation invariant and $\nien \in L^\zeta(\dd\mu)$, then $E_0 \in L^\zeta(\dd\mu)$.
\end{lemma}

\begin{proof}
    This is immediate from the assumptions and definitions.
\end{proof}

\begin{definition}[States in $\mathcal{I}_2$]
\label{def:I2}
    Let $\mathfrak{M}_\zeta(\mu) := \int |\nien|^\zeta \dd\mu$. A translation-invariant state~$\mu$ is said to belong to $\mathcal{I}_2$ if $\mathfrak{M}_\zeta(\mu) < \infty$ for some $\zeta > \max\{1, \tfrac12 \nu \}$.
\end{definition}

The following lemma is a simple adaptation of a computation done e.g.\ in~\cite[\S{4}]{LLL77}, and its corollary elucidates the use of the subscript ``$_2$'' in ``$\mathcal{I}_2$''.

\begin{lemma}
\label{lem:supp-on-Omega0}
    If $\mu \in \mathcal{I}$ is such that $\mathfrak{M}_\zeta(\mu) < \infty$ 
    for some $\zeta>0$, then, for every $r>\nu/\zeta$, $\mu$ is concentrated on~$\mathfrak{B}_{r}$ in the sense that $\mu(\mathfrak{B}_{r})=1$. 
\end{lemma}

\begin{proof}
    Let $r > 0$ be arbitrary. By Markov's inequality and translation invariance,
    \begin{align*}
        \mu\left\{(\bq,\bp) : \nieni(\bq,\bp) > (1+|i|)^r\right\}
        &\leq \frac{\mathfrak{M}_\zeta(\mu)}{(1+|i|)^{r\zeta}}.
    \end{align*}
    Provided that $\zeta r > \nu$, we can apply the Borel--Cantelli lemma to deduce that, for $\mu$-almost all~$(\bq,\bp)$, there exists $\Lambda \Subset \zz^\nu$ such that
    \[
        \sup\left\{ \frac{ \nieni(\bq,\bp)}{(1+|i|)^r}: i \notin\Lambda \right\} \leq 1.
    \]
    Therefore, 
    for $\mu$-almost all~$(\bq,\bp)$, there exists $C$ such that
    $
        \nieni(\bq,\bp) \leq C (1+|i|)^r
    $
    for all~$i\in\zz^\nu$. 
\end{proof}

\begin{corollary}
\label{cor:important-I2}
    Suppose (\ref{it:F}), (\ref{it:P1}), (\ref{it:D2}) and (\ref{it:TI}) hold. If $\mu \in \mathcal{I}_2$, then $\mu$ is concentrated on $\Omega_2$ and $E_0 \in L^\zeta(\dd\mu)$ for some $\zeta > 1$.
\end{corollary}

\begin{proof}
    The concentration on $\Omega_2$ comes directly from Lemma~\ref{lem:supp-on-Omega0}, since we can choose $r\in(0,2)$ such that $r>\nu/\zeta$. The fact that $E_0 \in L^\zeta(\dd\mu)$ for some $\zeta > 1$ follows from Lemma~\ref{lem:ien-from-nien}.
\end{proof}

\begin{remark}
    For the sake of simplicity and cohesion, we will often state results for configurations in~$\Omega_2$ or states in~$\mathcal{I}_2$ even though that is not necessarily the optimal assumption.
\end{remark}

We define the \emph{total energy in the finite subset} $\Lambda\Subset\zz^\nu$ as the function
\[
H_\Lambda\coloneqq \sum_{i\in\Lambda}K_{\{i\}} + \sum_{\Set \subseteq \Lambda} W_{\Set}
\]
---\,note that we are not using $\sum_{i\in\Lambda} E_i$. 
The next two lemmas provide useful relations between the local energy and the total energy in a finite subset.
\begin{lemma}
\label{lem:ham-from-nien}
     Suppose (\ref{it:F}), (\ref{it:P1}), (\ref{it:D2}) and (\ref{it:TI}) hold. 
     If $\nien \in L^1(\dd\mu)$, then $H_\Lambda \in L^1(\dd\mu)$ for every $\Lambda \Subset \zz^\nu$.
\end{lemma}

\begin{proof} 
    Since $L^1(\dd\mu)$ is a vector space, since~(\ref{it:TI}) holds, and since $\mu$ is translation invariant, it suffices to note that the assumption $\nien \in L^1(\dd\mu)$ implies $K_{\{0\}} \in L^1(\dd\mu)$ and $W_{\{0\}} \in L^1(\dd\mu)$ by~(\ref{it:P1}), as well as $W_\Set \in L^1(\dd\mu)$ for every $\Set \ni 0$ by~(\ref{it:D2})\,---\,by~(\ref{it:F}) there are finitely many such $\Set$, less than $2^{(2D+1)^\nu-1}-1$. Then, again by~(\ref{it:F}), we also have $H_\Lambda \in L^1(\dd\mu)$ for every $\Lambda \Subset \zz^\nu$. 
\end{proof}

\begin{lemma}
\label{lem:limitHLambda}
      Suppose (\ref{it:F}), (\ref{it:P1}), (\ref{it:D2}) and (\ref{it:TI}) hold. 
      If $\mu$ is a translation-invariant state with $\nien \in L^1(\dd\mu)$, then
    \begin{align*}
        \lim_{a\to\infty} \int  \frac{H_{\Lambda(a)}}{\#\Lambda(a)}\dd\mu = \int E_0 \dd\mu,
    \end{align*}
    with both sides being finite.
\end{lemma}

\begin{proof}
    By Lemma~\ref{lem:ham-from-nien}, $H_\Lambda \in L^1(\dd\mu)$ for every $\Lambda \Subset \zz^\nu$. 
    Let $a > \irange$ be arbitrary.
    For the purpose of the computation below, we temporarily merge the terms $K_{\{i\}}$ and $W_{\{i\}}$ into a term that we still call $W_{\{i\}}$. With this slight abuse of notation:
    \begin{align*}
        \int \frac{H_{\Lambda(a)}}{\#\Lambda(a)}\dd\mu
        &= \frac{1}{\#\Lambda(a)} \int \sum_{\Set \subseteq  \Lambda(a)} W_{\Set} \dd\mu \\
        &= \frac{1}{\#\Lambda(a)} \sum_{i \in \Lambda(a)}\int \sum_{\substack{\Set \ni i \\ \Set \subseteq  \Lambda(a)}} \frac{W_{\Set}}{\#\Set} \dd\mu \\
        &= \frac{1}{\#\Lambda(a)} \sum_{i \in \Lambda(a)\setminus\partial^{\textnormal{int}}\Lambda(a)}\int \sum_{\Set \ni i } \frac{W_{\Set}}{\#\Set} \dd\mu + \frac{1}{\#\Lambda(a)} \sum_{i \in \partial^{\textnormal{int}}\Lambda(a)}\int \sum_{\substack{\Set \ni i \\ \Set \cap \Lambda^\complement(a) = \emptyset}} \frac{W_{\Set}}{\#\Set} \dd\mu \\
        &= \frac{1}{\#\Lambda(a)} \sum_{i \in \Lambda(a)} \int \sum_{\Set \ni i } \frac{W_{\Set}}{\#\Set} \dd\mu - \frac{1}{\#\Lambda(a)} \sum_{i \in \partial^{\textnormal{int}}\Lambda(a)}\int \sum_{\substack{\Set \ni i \\ \Set \cap \Lambda^\complement(a) \neq \emptyset}} \frac{W_{\Set}}{\#\Set} \dd\mu.
    \end{align*}
    By translation invariance of $\mu$ and~(\ref{it:TI}), the first term equals $\int E_0 \dd\mu$. As for the second term, we can also use translation invariance of~$\mu$ and~(\ref{it:TI}) to obtain
    \begin{align*}
        \left|\frac{1}{\#\Lambda(a)} \sum_{i \in \partial^{\textnormal{int}}\Lambda(a)}\int \sum_{\substack{\Set \ni i \\ \Set \cap \Lambda^\complement(a) \neq \emptyset}} \frac{W_{\Set}}{\#\Set} \dd\mu\right|
        \leq \frac{\#\partial^{\textnormal{int}}\Lambda(a)}{\#\Lambda(a)} \sum_{\Set \supsetneq  \{0\}} \int |W_{\Set}| \dd\mu.
    \end{align*}
    Since the integral on the right-hand side is finite by~(\ref{it:D2}) and $\nien \in L^1(\dd\mu)$, and since $\#\partial^{\textnormal{int}}\Lambda(a) = o(\#\Lambda(a))$ as $a \to \infty$ by~(\ref{it:F}), we have the desired identity.
\end{proof}

\begin{remark}
    In particular, Lemmas~\ref{lem:ham-from-nien} and~\ref{lem:limitHLambda} hold when $\mu\in \mathcal I_2$ by 
    Jensen's inequality.
\end{remark}

We define the \emph{pressure at inverse temperature $\beta$ in the finite subset} $\Lambda\Subset\zz^\nu$ as
\begin{align*}
    P_\beta(H_\Lambda) := \ln \int \Exp{-\beta H_{\Lambda}} \dd\vol_{\Lambda}.
\end{align*}
By H\"older's inequality, the function $\beta \mapsto P_\beta(H_\Lambda)$ is convex.
The \emph{pressure per unit volume at inverse temperature $\beta$}, defined as
\begin{align*}
    p_\beta(H) := \limsup_{a\to\infty} \frac{P_\beta(H_{\Lambda(a)})}{\#\Lambda(a)},
\end{align*}
is then also convex in~$\beta$.

\begin{remark}
\label{rem:split-press}
    The pressure per unit volume $p_\beta(H)$ is a sum of two terms: one that is defined by keeping only the potential terms in $H_\Lambda$ (both pinning and interaction), and one that is defined by keeping only the kinetic terms in $H_\Lambda$. Both are convex functions of $\beta$, and the latter is proportional to $-\ln \beta$ and therefore strictly convex and smooth in~$\beta > 0$. In particular, the pressure $p_\beta(H)$ is strictly convex where it is finite, and its regularity is determined by the potential terms only.
\end{remark}

\begin{lemma}
\label{lem:VP-upper-bound}
    Suppose~(\ref{it:F}), (\ref{it:P1}), (\ref{it:P4}),\footnote{The only property we actually use is the strong pinning condition $\Exp{-c W_{\{i\}}} \in L^1(\dd\vol)$ for all $i$ and all $c>0$, which is implied by~(\ref{it:P4}).} (\ref{it:D2}) and (\ref{it:TI}) hold. If $\mu \in \mathcal{I}_2$, then 
    \begin{equation}
        s(\mu) - \beta \int E_0 \dd\mu \leq p_\beta(H)
    \end{equation}
    for every $\beta > 0$.
\end{lemma}

\begin{proof}
    By rescaling $H_\Lambda$, it suffices to consider the case $\beta = 1$.
    Observe that $E_0 \in L^1(\dd\mu)$ by Corollary~\ref{cor:important-I2}, so the statement trivially holds if $s(\mu) = -\infty$ or $p(H) = \infty$. On the other hand, Proposition~\ref{prop:s-prop} guarantees that $s(\mu) < \infty$; see Remark~\ref{rem:concrete-frak-F}.
    Hence, in view of Proposition~\ref{prop:s-prop}, we may safely work with finite $S_{\Lambda(a)}(\mu\res{\Lambda(a)})$ for all $a\in\nn$. 
    
    Now,
    since $H_{\Lambda(a)} \in L^1(\dd\mu)$ by Lemma~\ref{lem:ham-from-nien}, we have
    \begin{align*}
        S_{\Lambda(a)} (\mu\res{\Lambda(a)}) - \int H_{\Lambda(a)} \dd\mu
        &\leq  P(H_{\Lambda(a)})
    \end{align*}
    for all $a\in\nn$ by Lemma~\ref{lem:rel-KL}; see Remark~\ref{rem:var-ineq}.
    Dividing by $\#\Lambda(a)$ and taking the limit superior as $a \to \infty$, the desired conclusion follows from Lemma~\ref{lem:limitHLambda} and Proposition~\ref{prop:s-prop}.
\end{proof}

\begin{remark}
    See Lemma~\ref{lem:existence-eq} for a result on the existence of a state in $\mathcal{I}_2$ that saturates this upper bound.
\end{remark}

\subsection{The dynamics}
\label{ssec:dyn}

Formally, in the canonical coordinates $(\bq,\bp)$, the dynamics should be
\begin{subequations}
\label{eq:dyn}
\begin{align}
    \dot q_i &= p_i, \\ 
    \dot p_i &= -\sum_{{\Set} \Subset \zz^\nu} \partial_{q_i} W_{\Set}(\bq).
\end{align}
\end{subequations}
As is customary (see e.g.~\cite{LLL77}), we tackle existence and uniqueness of solutions by severing the dynamics at the boundary of~$a$-elementary boxes:\footnote{We consider a slight variant of~\cite[Eq.~(9)]{LLL77}, motivated by Theorem~\ref{thm:conservation-entropy}. The original condition $\Set\supseteq \{i\}$ for $i\in\Lambda(a)$, without the constraint $\Set\subseteq \Lambda(a)$, would require the additional assumption~(\ref{it:F}) in Lemma~\ref{lem:existUnique}.}
\begin{subequations}
\label{eq:a-severed-dyn}
\begin{align}
    \dot q^a_i &= p^a_i, \\ 
    \dot p^a_i &= -\sum_{\substack{\Set \subseteq  \Lambda(a)}} \partial_{q_i} W_{\Set}(\bq^a),
\end{align}
\end{subequations}
for $i \in \Lambda(a)$, the existence and uniqueness theory for those being standard. 

\begin{lemma}[Existence and uniqueness of time evolution for the severed dynamics]\label{lem:existUnique}
    Let $\Lambda(a) \Subset \zz^\nu$ be an $a$-elementary box, and consider the severed dynamics~\eqref{eq:a-severed-dyn} with initial conditions $(\bq^a(0),\bp^a (0))\in \prod_{i\in\Lambda(a)}T^*X_i$.
    Under assumptions (\ref{it:R}), (\ref{it:P1})--(\ref{it:P3})\footnote{As far as~(\ref{it:P3}) is concerned, we actually only use the compactness of the sublevel sets.}, (\ref{it:D1}) and (\ref{it:TI}), the system has a unique global-in-time solution $(\bq^a(t),\bp^a (t))\in C^1(\rr,\bigtimes_{i \in \Lambda(a)} T^*X_i)$. Moreover, there exists a bounded, translation-invariant band-matrix $A$ on $\ell^2(\zz^\nu)$, such that these solutions satisfy
    \begin{equation}
    \label{eq:a-priori-in-box}
        \nieni(\bq^a(t),\bp^a(t)) \leq \sum_{j \in \Lambda(a)}  [\Exp{tA}]_{i,j} \nienj(\bq^a(0),\bp^a(0))
    \end{equation}
    for every $a\in\nn$, $i \in \Lambda(a)$ and $t\geq 0$.
\end{lemma}
\begin{proof}
    We follow the standard strategy, emphasizing the a priori bounds that will be useful for the infinite-volume dynamics.
    Local (short time) existence and uniqueness follows from the Cauchy--Lipschitz theorem in finite dimension, thanks to the regularity assumption~(\ref{it:R}).
        
    To show that the solution cannot blow up in finite time, we establish a minor variant of an a priori estimate from~\cite[\S{2}]{LLL77}.
    Let $\mathcal{L}^a_i(t)\coloneqq \nieni(\bq^a(t),\bp^a(t)) $ for all $i\in\Lambda(a)$ and let $\mathcal L ^a$ the corresponding finite array.
    A straightforward computation using~\eqref{eq:a-severed-dyn} gives
    \begin{align*}
        \frac{\dd}{\dd t} \mathcal{L}^a_i(t) 
        & =-p_i^a(t)\sum_{\substack{\Set \subseteq  \Lambda(a)\\
    \Set\supsetneq \{i\}}} \partial_{q_i} W_{\Set}(\bq^a(t)).
    \end{align*}
    Using Cauchy's inequality and~(\ref{it:D1}), we then obtain
    \begin{align*}
    \frac{\dd}{\dd t} \mathcal{L}^a_i(t) 
        & \leq K_{\{i\}}^a +\frac 1 2 \left| \sum_{\substack{\Set \subseteq  \Lambda(a)\\
       \Set\supsetneq \{i\}}} \partial_{q_i} W_{\Set}(\bq^a(t))\right|^2\\
        & 
        \leq\sum_{j\in\Lambda(a)}A_{i,j}  \mathcal{L}^a_j(t) ,
    \end{align*}
    where $A_{i,j}\coloneqq\max\{1,C_1\}$ whenever $i,j\in\Lambda(a)$ with $\gamma(i,j)\leq 1$, and $A_{i,j}\coloneqq 0$ otherwise.
    Therefore, with $A$ the linear induced by $(A_{i,j})_{i,j \in \Lambda(a)}$, we have
    \begin{equation}
    \label{eq:a-priori-est}
        \mathcal{L}^a_i(t) \leq \left(\Exp{t A}\mathcal{L}^a(0)\right)_i.
    \end{equation}
    Hence, $\mathcal{L}^a_i(t) < \infty$ for all finite $t$.
    Since $W_{\{i\}}(\bq^a(t))\leq \mathcal{L}^a_i(t) $ by~(\ref{it:P1}), $\bq^a(t)$ stays in a compact set for all $t\in(0,\infty)$  by~(\ref{it:P3}). The same directly follows for $\bp^a(t)$.
    Since the solution $(\bq^a(t),\bp^a(t))$ remains bounded, it can be continued beyond $t$. According to the local theory, uniqueness still applies.
\end{proof}

\begin{lemma}[Existence and uniqueness of time evolution]
\label{lem:dyn-bounds}
    Suppose that (\ref{it:F}), (\ref{it:R}), (\ref{it:P1})--(\ref{it:P4}), (\ref{it:D1})--(\ref{it:D3}) and  (\ref{it:TI}) hold and let $r \in (0,2)$. For every initial condition $(\bq,\bp) \in \mathfrak{B}_{r}$ and every time interval~$[0,T]$, there exists a unique solution to~\eqref{eq:dyn} that remains in that~$\mathfrak{B}_{r}$. 
    
    Moreover, there exists a constant $\alpha$ with the following property: for every initial condition $(\bq,\bp) \in \mathfrak{B}_{r,C}$,
    the above solution $t\mapsto (\bq(t),\bp(t))$ to~\eqref{eq:dyn} satisfies
    \begin{align*}
        \nieni(\bq(t),\bp(t))
        \leq C \Exp{\alpha t} (1+|i|)^r
    \end{align*}
    and there exists a continuous, monotone function $t\mapsto c_t$ such that, whenever $i \in \Lambda(a)$ and $t\geq 0$, we have
   \begin{align*}
    |q_i^a(t) - q_i(t)| 
    &\leq  \frac{c_t^{\gamma+1}}{(2\gamma)!} (1+|i|+\gamma\irange)^{r(\gamma+1)}\\
    |p_i^a(t) - p_i(t)|
    &\leq  \frac{c_t^{{\gamma+1}}}{(2{(\gamma-1)})!} (1+|i|+\gamma\irange)^{r(\gamma+1)}
    \end{align*}
    where $\gamma \leq \gamma(i,\Lambda(a)^\complement)$ and $t\mapsto  (\bq^a(t),\bp^a(t))$ is the corresponding solution  to~\eqref{eq:a-severed-dyn}.
\end{lemma}

\begin{proof}
    First, let us exhibit an a priori estimate along the lines of~\eqref{eq:a-priori-in-box}.
    To this end, we consider a putative solution $t \mapsto (\bq^{\#}(t),\bp^{\#}(t))$ to~\eqref{eq:dyn} 
    or~\eqref{eq:a-severed-dyn},
    let $\mathcal{L}^{\#}_i(t) := \nieni(\bq^{\#}(t),\bp^{\#}(t))$, 
    and let $\mathcal{L}^{\#}$ be the corresponding infinite array. 
    Here, the symbol ``$\#$'' is used as a placeholder for either some $a \in \nn$ or nothing (i.e.\ $a=\infty$).
    It follows from the same argument as for~\eqref{eq:a-priori-in-box} that
    \begin{equation}
    \label{eq:a-priori-from-G}
        \|\mathcal{L}^{\#}(t)\|_r \leq \Exp{t\|A\|_{\textnormal{op},r}} \|\mathcal{L}^{\#}(0)\|_r,
    \end{equation}
    where, by~(\ref{it:F}), $A$ is a bounded operator on the space~$(\mathfrak{A}_{r},\|\,\cdot\,\|_r)$ described in Section~\ref{ssec:spaces}.
    
    The basic idea is the following. As in~\cite[\S{2}]{LLL77}, 
    Lemma~\ref{lem:existUnique} and the a priori estimate~\eqref{eq:a-priori-from-G} imply existence of a solution in~$\mathfrak{B}_r$ to~\eqref{eq:dyn} for any initial condition in~$\mathfrak{B}_r$ (this will actually be reproved below).
    Thanks to (\ref{it:D2})--(\ref{it:D3}), if $r>0$ is small enough, then for every $C>0$, the assignment $\zz^\nu \ni i \mapsto C(1+(|i|))^r$ is a ``sequence of uniqueness'' in the sense of~\cite[\S{3}]{LLL77}. Hence, provided that $r$ is small enough, we have existence and uniqueness of a solution in~$\mathfrak{B}_r$ to~\eqref{eq:dyn} for any initial condition in~$\mathfrak{B}_r$ (this will actually be reproved below).

    {
    Fix $i\in\zz^\nu$. We first want to show that the sequence $(q_i^a (t))_{a \in \nn}$ is Cauchy in $a$. When we compare the solutions to the dynamics severed at scales $a\leq b$, with the same initial conditions in~$\mathfrak{B}_{r,C}$, we find
\begin{align*}
    | q_i^a(t) -  q_i^b(t)|
    &\leq \int_0^t (t-t_1) \left|\sum_{\Delta \ni i} \partial_{q_i} W(\bq^{a}(t_1)) - \partial_{q_i} W(\bq^{b}(t_1))\right| \dd t_1.
\end{align*}
To control the difference in forces, we interpolate between $\bq^{a}(t_1)$ and $\bq^{b}(t_1)$ using a chain of one-dimensional paths along (finitely many) individual sites. Without loss of generality, we temporarily and arbitrarily order the sites $i\in\zz^\nu$. For each site $j$, define a path $ u \mapsto \mathbf{x}^{(a,b,t_1,j)}(u)$ that linearly interpolates between $\bq^{a}(t_1)$ and $\bq^{b}(t_1)$ at site $j$ as $u$ ranges over $[0,1]$, while holding other sites fixed:
\begin{align*}
    [\mathbf{x}^{(a,b,t_1,j)}(u)]_{j'} = 
    \begin{cases}
        q_{j'}^{a}(t_1) & j' < j, \\
        (1-u)q_{j}^{b}(t_1) + u q_{j'}^{a}(t_1) & j' = j, \\
        q_{j'}^{b}(t_1) & j' > j. \\
    \end{cases}
\end{align*}
Using the fundamental theorem of calculus and~(\ref{it:F}), we obtain
\[
\partial_{q_i} W(\bq^{a}(t_1)) - \partial_{q_i} W(\bq^{b}(t_1))= \sum_{j:\gamma(i,j)\leq 1}(q^{a}_j(t_1)-q^{b}_j(t_1))\int_0^1 \partial_{q_j}\partial_{q_i}W(\mathbf{x}^{(a,b,t_1,j)}(u))\dd u,
\]
and by Assumption~(\ref{it:D3}) we bound the second derivatives of the interaction energy in terms of local on-site potentials 
\begin{align*}
    \left|\sum_{\Delta \ni i} \partial_{q_i} W(\bq^{a}(t_1)) - \partial_{q_i} W(\bq^{b}(t_1))\right|
    &\leq \sum_{j : \gamma(i,j)\leq 1} |q_j^a(t_1) - q_j^b(t_1)|\int_0^1 \left|\sum_{\Delta \ni i} \partial_{q_j}\partial_{q_i}W_\Delta(\mathbf{x}^{(a,b,t_1,j)}(u))\right| \dd u \\
    &\leq \sum_{j : \gamma(i,j)\leq 1} |q_j^a(t_1) - q_j^b(t_1)| \sup_{u \in [0,1]} C_3 \sum_{j' : \gamma(i,j') \leq 1} W_{\{j'\}}(\mathbf{x}^{(a,b,t_1,j)}(u)).
\end{align*}
Now using the a priori estimate on $\bq^a(t_1)$ and $\bq^b(t_1)$, the definition of $\mathbf{x}^{(a,b,t_1,j)}$, the convexity-like property~(\ref{it:P3}) to assert that the path remains in a region where the on-site potential is controlled, and the finite-range assumption~(\ref{it:F}), we find
\begin{align*}
    \left|\sum_{\Delta \ni i} \partial_{q_i} W(\bq^{a}(t_1)) - \partial_{q_i} W(\bq^{b}(t_1))\right|
    &\leq \sum_{j : \gamma(i,j)\leq 1} |q_j^a(t_1) - q_j^b(t_1)| C(t) \max_{j' : \gamma(i,j') \leq 1} (1+|j'|)^r\\
    &\leq C(t) (1+|i|+D)^r\sum_{j : \gamma(i,j)\leq 1} |q_j^a(t_1) - q_j^b(t_1)|,
\end{align*}
where $C(t) = C_3 C \Exp{t\|A\|_{\textnormal{op},r}} (2D+1)^\nu$ with $C$ from the a priori energy bound~\eqref{eq:B-en}.

\noindent
Iterating $\gamma$ times, where $\gamma \leq \gamma(i,\Lambda(a)^\complement)$, and using Assumption~(\ref{it:P4}) with~\eqref{eq:a-priori-from-G}, we find
\begin{align*}
    |q_i^a(t) - q_i^b(t)|
    &\leq C(t)^\gamma (1+|i|+D)^r (1+|i|+2D)^r \dotsb (1+|i|+(\gamma-1)D)^r (1+|i|+\gamma D)^r \\
    &\qquad \int_0^t (t-t_\gamma)\int_0^{t_\gamma} (t_\gamma-t_{\gamma-1}) \dotsb \int_0^{t_2} (t_2-t_1) {\sup_{j_\gamma : \gamma(i,{j_\gamma}) \leq \gamma} |q_{j_\gamma}^a(t_1) - q_{j_\gamma}^b(t_1)|} \dd t_1 \dd t_2 \dotsc \dd t_\gamma \\
    &\leq C(t)^\gamma (1+|i|+\gamma D)^{r\gamma}  {2 C_0 C \Exp{\|A\|_{\textnormal{op},r} t} (1+|i|+\gamma D)^r} \frac{t^{2\gamma}}{(2\gamma)!}.
\end{align*}
With $C_0,C, \|A\|_{\textnormal{op},r},D,|i|,r$ being fixed, and with $t$ in a bounded interval~$[0,T]$, this last expression is $\exp\left( r\gamma \log(\gamma)-2\gamma \log(2\gamma)+O(\gamma)\right)$
as $\gamma \to \infty$, that is $a \to \infty$.} Since $r \in (0,2)$, we can deduce that $(q_i^a(\,\cdot\,))_{a \in \nn}$ is a Cauchy sequence in $C([0,T],X_i)$, and thus admits a limit $q_i^\infty(\,\cdot\,)$. Since $i$ was arbitrary, we conclude that this is the case for all~$i\in\zz^\nu$.

    Note that the above estimates only depend on the size of~$a$,~$t$ and~$|i|$. 
    Then, each limit $q_i^\infty(\,\cdot\,)$ satisfies
    \begin{align*}
        q_i^\infty(t) 
            &= \lim_{a\to\infty} \left(q_i^a(0) + t p_i^a(0) - \sum_{\Set \ni i}\int_0^t (t-t_1) \partial_{q_{i}}W_{\Set}(\bq^a(t_1)) \dd t_1 \right) \\
            &= q_i^\infty(0) + t p_i^\infty(0) - \sum_{\Set \ni i}\int_0^t (t-t_1) \partial_{q_{i}}W_{\Set}\left(\lim_{a\to\infty} \bq^a(t_1)\right) \dd t_1 \\
            &= q_i^\infty(0) + t p_i^\infty(0) - \sum_{\Set \ni i}\int_0^t (t-t_1) \partial_{q_{i}}W_{\Set}(\bq^\infty(t_1)) \dd t_1,
    \end{align*}
    i.e.\ defines a solution in~$\mathfrak{B}_r$. 
    
    Since the estimates used for two solutions to the system~\eqref{eq:a-severed-dyn} severed at different scales at least~$a$ also apply to two putative solutions in~$\mathfrak{B}_{r,C}$ to the original dynamics~\eqref{eq:dyn}, we can also deduce uniqueness. 

    Since these same estimates also apply to a solution to the system~\eqref{eq:a-severed-dyn} severed at scale~$a$ and the solution to the original dynamics~\eqref{eq:dyn}, we can also deduce the proposed bounds for the position variables.
    To bound the momentum variables, we use
    \begin{align*}
        p_i^a(t) - p_i^b(t)
        &= \int_0^t \sum_{\Set \ni i}  \partial_{q_{i}}W_{\Set}(\bq^b(t_1)) - \partial_{q_{i}}W_{\Set}(\bq^a(t_1))  \dd t_1
    \end{align*}
    and the above arguments; we find
    \begin{align*}
        |p_i^a(t) - p_i^b(t)|
        &\leq \frac{t^{2{(\gamma-1)+1}} C(t)^{\gamma}}{(2{(\gamma-1)})!} 2C_0 C \Exp{\|A\|_{\textnormal{op},r} t}(1+|i|+\gamma\irange)^{r(\gamma+1)}.
    \end{align*}
    The proposed bounds are equivalent up to relabelling of the function $t \mapsto c_t$.
\end{proof}

\begin{proposition}[Time-reversal-invariant dynamical group]
\label{prop:TRI}
    The equations~\eqref{eq:dyn} define a dynamical group $(\tau_t)_{t\in\rr}$ of maps $\tau_t : \Omega_2 \to \Omega_2$. Moreover, with $\vartheta(\bq,\bp) := (\bq,-\bp)$, we have $\tau_t^{-1} = \vartheta \circ \tau_t \circ \vartheta$ for all~$t \in \rr$. 
\end{proposition}

\begin{proof}
    The definition of the group is unambiguous from Lemma~\ref{lem:dyn-bounds} and the group property for $(\tau_t^a)_{t\in\rr}$ at fixed $a\in\nn$. 
    It is easy to see from the definitions that $\Omega_2$ is invariant under~$\vartheta$.
    Since the symmetry $(\tau_t^a)^{-1} = \vartheta \circ \tau_t^a \circ \vartheta$ holds for all~$t \in \rr$ given any fixed~$a \in \nn$, the result follows from Lemma~\ref{lem:dyn-bounds}.
\end{proof}

We may at times consider $\tau_t$ applied to a configuration in possible sets of measure 0 outside of~$\Omega_2$. To deal with this technicality, we set $\tau_t(\bq,\bp) = (\bq,\bp)$ for $(\bq,\bp) \notin \Omega_2$.

\section{Conservation of specific energy}
\label{sec:energy}

Given an extended state~$\mu$ and $t \in \rr$, we define $\mu_t\coloneqq \mu \circ \tau_t^{-1}$ where $t \mapsto \tau_t(\bq,\bp)$ is the solution with initial condition $(\bq,\bp)$ that was identified in Proposition~\ref{prop:TRI}. In particular, $\mu_0 = \mu$. 
The main result of this section is that the mean local energy-per-particle under $\mu_t$ is constant in~$t$. But first, we state and prove a lemma about preservation of finiteness. 

\begin{lemma}
\label{lem:pre-energy-cons}
    Suppose that  (\ref{it:F}), (\ref{it:R}), (\ref{it:P1})--(\ref{it:P4}), (\ref{it:D1})--(\ref{it:D3}) and  (\ref{it:TI}) hold. If $\mu$ is a translation-invariant state such that $\mathfrak{M}_\zeta(\mu)<\infty$ for some $\zeta > \tfrac 12 \nu$, then $\mathfrak{M}_{\xi}(\mu_t)<\infty$ for all $\xi \in (0,\zeta)$ and $t \in \rr$. Moreover, the growth is at most exponential in~$t$.
\end{lemma}

\begin{proof}
    By Lemma~\ref{lem:supp-on-Omega0}, our assumption that $\mathfrak M_\zeta(\mu)<\infty$
    for some $\zeta > \tfrac 12 \nu$ implies that $\mu$ is supported on~$\mathfrak{B}_{r}$ for all $r$ close enough to 2.
    For every $t >0$, we know from Lemma~\ref{lem:dyn-bounds} that
    \begin{align*}
    (\bq,\bp) \in \mathfrak{B}_{r,C}
        &\implies (\bq(t),\bp(t)) \in \mathfrak{B}_{r,C\Exp{\alpha_r t}}.
    \end{align*}
    Conversely,
    \begin{align*}
        (\bq,\bp) \in \mathfrak{B}_r \text{ and } |\nien(\bq(t),\bp(t))|^\xi > \varepsilon
        &\implies 
        (\bq,\bp) \notin \mathfrak{B}_{r, \varepsilon^{\frac 1\xi} \Exp{-\alpha_r t}}.
    \end{align*}
    Therefore, by standard probabilistic arguments and the fact that $\tau_t^{-1}(\mathfrak{B}_r)=\mathfrak{B}_r$ by Theorem~\ref{lem:dyn-bounds}, 
    \begin{align*}
        \int_{\mathfrak{B}_r} |\nien|^\xi \dd\mu_t
        & =  \int_{\tau_t^{-1}(\mathfrak{B}_r)} |\nien\circ \tau_t|^\xi \dd\mu\quad\\
        & =\int_0^\infty \mu \left\{ (\bq,\bp) \in \mathfrak{B}_r : |\nien(\bq(t),\bp(t))|^\xi > \varepsilon \right\}\dd \varepsilon\\
        &\leq  1 + \int_1^\infty \mu((\mathfrak{B}_{r,\ \varepsilon^{\frac 1\xi} \Exp{-\alpha_r t}})^\complement) \dd \varepsilon  \\
        &\leq 1 + \int_1^\infty \sum_{{i} \in \zz^\nu} \frac{\int |\nieni|^\zeta \dd\mu}{(\varepsilon^{\frac 1\xi} \Exp{-\alpha_r t})^\zeta (1+|i|)^{\zeta r}} \dd \varepsilon \\
    &\leq 1 +   \Exp{\alpha_r\zeta t} \mathfrak{M}_\zeta(\mu) \left(\sum_{{i} \in \zz^\nu} \frac{1}{(1+|i|)^{\zeta r}}\right) \int_1^\infty  \varepsilon^{-{\frac{\zeta}{\xi}}} \dd \varepsilon.
    \end{align*}
    The right-hand side will be finite independently of~$t$, provided that $\zeta r > \nu$ and $\tfrac{\zeta}{\xi} > 1$. We conclude by taking $r \to 2$.
\end{proof}

\begin{remark}\label{rmk:MomSevDyn}
    Note that the conclusion of Lemma~\ref{lem:pre-energy-cons} also holds for the family of measures
    \begin{equation}\label{eq:mu_t_severed}
        \mu_t^a \coloneqq \mu \circ (\tau_t^a)^{-1},
    \end{equation}
    where $t \mapsto \tau_t^a(\bq, \bp)$ denotes the solution to the severed dynamics~\eqref{eq:a-severed-dyn} with initial condition $(\bq, \bp)$. That is, for any $\xi \in (0, \zeta)$ and $t \in \mathbb{R}$, we have $\mathfrak{M}_\xi(\mu_t^a) < \infty$, with exponential control in $t$, under the same assumptions as in Lemma~\ref{lem:pre-energy-cons}.
\end{remark}

\begin{corollary}
\label{cor:pre-energy-cons}
    Under the assumptions of Lemma~\ref{lem:pre-energy-cons},
    the set $\mathcal{I}_2$ is invariant under $\tau_t$ and
    \[ 
        \int E_0 \dd\mu_t < \infty
    \]
    for all $\mu \in \mathcal{I}_2$ and $t \in \rr$.
\end{corollary}

\begin{proof}
    Let $t \in \rr$ be arbitrary. If $\mu \in \mathcal{I}_2$, then there exists $\zeta > \max\{1,\tfrac 12 \nu\}$ such that $\mathfrak{M}_\zeta(\mu) < \infty$. By Lemma~\ref{lem:pre-energy-cons}, this implies $\mathfrak{M}_{\xi}(\mu_t)<\infty$ for some $\xi > \max\{1,\tfrac 12 \nu\}$. Therefore, $\mu_t \in \mathcal{I}_2$. Integrability of $E_0$ with respect to $\mu_t$ then follows from Lemma~\ref{lem:ien-from-nien}.
\end{proof}

\begin{proposition}[Conservation of energy under time evolution]
\label{prop:energy-conservation}
    Under the assumptions of Lemma~\ref{lem:pre-energy-cons}, if $\mu\in\mathcal{I}_2$, then 
    \begin{align*}
        \int E_0 \dd\mu = \int E_0 \circ \tau_t \dd\mu
    \end{align*}
    for all~$t \in \rr$.
\end{proposition}

\begin{proof} 
    Throughout this proof, we use $H$ as a shorthand for $H_\Lambda$ with $\Lambda$ chosen large enough to contain all sets $\Set$ that include $0$ and have nontrivial interaction terms. 
    Using the integrability provided by Corollary~\ref{cor:pre-energy-cons}, we see that
    \begin{align*}
        \int_{\Omega_2}E_0 \dd\mu_t - \int_{\Omega_2}E_0 \dd\mu 
        &= \int_{\Omega_2}(E_0 \circ\tau_t) \dd\mu - \int_{\Omega_2}E_0 \dd\mu \\
        &= \int_{\Omega_2}(E_0 \circ\tau_t) - E_0 \dd\mu \\
        &= \int_{\Omega_2}\int_0^t  \{E_0,H\} \circ \tau_s \dd s \dd\mu
    \end{align*}
    ---\,recall that $\mu$ is supported on~$\Omega_2$, which is invariant under the time evolution by Proposition~\ref{prop:TRI}.
    Hence, it suffices to show that 
    \begin{equation}
    \label{eq:alt-energy-cons}
        \int_{\Omega_2}\int_0^t \{E_0, H\} \circ \tau_s \dd s \dd\mu = 0
    \end{equation}
    for all~$t \geq 0$.
    
    Let us compute
    \begin{align}
        \left\{{E}_0, H\right\} 
        &= \sum_{\Set\subseteq \Lambda} \left\{ K_{\{0\}} , W_{\Set}\right\} + \sum_{x \in \Lambda} \left\{\sum_{\Set \ni 0} \frac{W_{\Set}}{\#\Set},  K_{\{x\}} \right\} \notag \\
        &= -\sum_{\Set \ni 0} p_0 \partial_{q_0}W_{\Set} + \sum_{\Set \ni 0} \frac{1}{\#\Set} \sum_{x \in {\Set}}  p_x \partial_{q_x}W_{\Set} \notag \\
        &=  \sum_{\Set\supsetneq \{0\}}\left( -p_0 \partial_{q_0}W_{\Set} + \frac{1}{\#\Set} \sum_{x \in {\Set}}  p_x \partial_{q_x}W_{\Set}\right).
        \label{eq:P-bracket}
    \end{align}
    This computation has two consequences:
    \begin{enumerate}
        \item {Applying the triangle inequality and Cauchy's inequality to~\eqref{eq:P-bracket}, Corollary~\ref{cor:pre-energy-cons}} can be used in conjunction with~(\ref{it:D1}) and (\ref{it:TI}) to show that
        \[ 
            \int_0^t \int_{\Omega_2} |\left\{{E}_0,H\right\}| \dd\mu_s \dd s < \infty
        \]
        for all~$t \geq 0$. 
        
        \item 
            We claim that the two sums in~\eqref{eq:P-bracket} cancel each other out when integrated against a translation-invariant measure $\nu$ for which 
            $\mathfrak{M}_\zeta(\nu) < \infty$ with $\zeta > \max\{1,\tfrac 12 \nu\}$. 
            On the one hand, we have  
            \begin{align*}
                \int_{\Omega_2}\sum_{\Set \ni 0} p_0 \partial_{q_0}W_{\Set} \dd \nu
                &=  \sum_{\substack{\Set_* \ni 0 \\ \min \Set_* = 0}} \sum_{z \in {\Set}_*} \int_{\Omega_2}p_0 \partial_{q_0}W_{\mathsf{T}_{-z}(\Set_*)} \dd \nu.
            \end{align*}
            We have simply rewritten the set $\Set$ containing~$0$ as a translate of a set $\Set_*$ that has~$0$ as its minimum (in lexicographical order), and used the fact that {there are finitely many $\Set$ by~(\ref{it:F}) and} all summands are integrable.
            On the other hand, the same arguments give
            \begin{align*}          &\int_{\Omega_2}\sum_{\Set \ni 0} \frac{1}{\#\Set} \sum_{x \in {\Set}}  p_x \partial_{q_x}W_{\Set} \dd \nu \\
                &\qquad\qquad = \sum_{\substack{\Set_* \ni 0 \\ \min \Set_* = 0}} \sum_{y \in {\Set}_*} \frac{1}{\#\Set_*} \sum_{x \in \mathsf{T}_{-y}(\Set_*)} \int_{\Omega_2} p_x \partial_{q_x}W_{\mathsf{T}_{-y}(\Set_*)} \dd \nu \\
                &\qquad\qquad = \sum_{\substack{\Set_* \ni 0 \\ \min \Set_* = 0}} \sum_{y \in {\Set}_*} \frac{1}{\#\Set_*} \sum_{x \in \mathsf{T}_{-y}(\Set_*)} \int_{\Omega_2} [p_0 \circ \mathsf{T}^*_{-x}] [(\partial_{q_x}W_{\mathsf{T}_{-y}(\Set_*)})\circ \mathsf{T}^*_{x}\circ \mathsf{T}^*_{-x}] \dd \nu.
            \end{align*}
            Using~(\ref{it:TI}), then using translation invariance of the measure~$\nu$, and finally reindexing, we find
            \begin{align*}
                &\int_{\Omega_2}\sum_{\Set \ni 0} \frac{1}{\#\Set} \sum_{x \in {\Set}}  p_x \partial_{q_x}W_{\Set} \dd \nu \\
                &\qquad = \sum_{\substack{\Set_* \ni 0 \\ \min \Set_* = 0}} \sum_{y \in {\Set}_*} \frac{1}{\#\Set_*} \sum_{x \in \mathsf{T}_{-y}(\Set_*)} \int_{\Omega_2} [p_0 \circ \mathsf{T}^*_{-x}] [(\partial_{q_0}W_{\mathsf{T}_{-x-y}(\Set_*)})\circ \mathsf{T}^*_{-x}] \dd \nu \\
                &\qquad = \sum_{\substack{\Set_* \ni 0 \\ \min \Set_* = 0}} \sum_{y \in {\Set}_*} \frac{1}{\#\Set_*} \sum_{x \in \mathsf{T}_{-y}(\Set_*)} \int_{\Omega_2} p_0\partial_{q_0} W_{\mathsf{T}_{-x-y}(\Set_*)} \dd \nu \\
                &\qquad = \sum_{\substack{\Set_* \ni 0 \\ \min \Set_* = 0}} \sum_{y \in {\Set}_*} \frac{1}{\#\Set_*} \sum_{z \in {\Set}_*} \int_{\Omega_2} p_0\partial_{q_0} W_{\mathsf{T}_{-z}(\Set_*)} \dd \nu \\
                &\qquad = \sum_{\substack{\Set_* \ni 0 \\ \min \Set_* = 0}} \sum_{z \in {\Set}_*} \int_{\Omega_2} p_0\partial_{q_0} W_{\mathsf{T}_{-z}(\Set_*)} \dd \nu .
            \end{align*}
            Therefore, $\int_{\Omega_2}\left\{ {E}_0, H\right\} \dd\nu = 0$
            for every translation-invariant state~$\nu$ for which $\mathfrak{M}_\zeta(\nu) < \infty$ for $\zeta > \max\{1,\tfrac 12 \nu\}$.
        \end{enumerate}
    By Corollary~\ref{cor:pre-energy-cons}, the second consequence applies to $\mu_s$ for all~$s\geq 0$, so 
    \begin{align*}
        0 
            &= \int_{\Omega_2}\left\{ {E}_0, H\right\} \dd\mu_s 
            = \int_{\Omega_2}\left\{ {E}_0, H\right\}\circ\tau_s \dd\mu
    \end{align*}
    for all~$s \geq 0$. Therefore, by the first consequence and the Fubini--Tonelli theorem,
    \begin{align*}
        0 
            &= \int_0^t \int_{\Omega_2}\left\{ {E}_0, H\right\}\circ\tau_s \dd\mu \dd s
            =  \int_{\Omega_2}\int_0^t \left\{ {E}_0, H\right\}\circ\tau_s \dd s \dd\mu 
    \end{align*}
    for all~$t\geq 0$, establishing~\eqref{eq:alt-energy-cons}.
\end{proof}

\section{Conservation of specific entropy}
\label{sec:entropy}

\begin{theorem}[Conservation of entropy under time evolution]
\label{thm:conservation-entropy}
    Suppose that (\ref{it:F}), (\ref{it:R}), (\ref{it:P1})--(\ref{it:P4}), (\ref{it:D1})--(\ref{it:D3}) and (\ref{it:TI}) hold. If $\mu\in\mathcal{I}_2$, then 
    \begin{equation}\label{eq:cons-entropy}
        s(\mu_t) = s(\mu)
    \end{equation}
    for every $t \in \rr$.
\end{theorem}

\begin{proof}
    We follow the strategy of Lanford and Robinson~\cite[\S{4}]{LR68}. 
    \begin{description}
        \item[Step 1: Reduction.] Because our assumptions on~$\mu$ and $H$ are symmetric in~$\bp$, time-reversal invariance (Proposition~\ref{prop:TRI}) implies that it suffices to show that, under the ongoing assumptions, $s(\mu_t) \geq s(\mu)$ for all $t \in (0,1]$. For this purpose, we may assume that $s(\mu) \neq -\infty$: indeed, if $s(\mu) = -\infty$, then trivially $s(\mu_t) \geq s(\mu)$, whether $s(\mu_t)$ is finite or not.
        
        Now, one can argue using upper semicontinuity (Lemma~\ref{prop:s-prop}) that it suffices to find, for every~$t \in (0,1]$, a sequence $(\overline{\mu}_t^a)_{a\in\nn}$ of translation-invariant states such that: 
        \begin{enumerate}
            \item[i.] the sequence $(\overline{\mu}_t^a)_{a\in\nn}$ satisfies the weak-convergence criterion to $\mu_t$ as $a\to\infty$ in Lemma~\ref{lem:topo-on-states};
            \item[ii.] for $a$ large enough, $s(\overline{\mu}_t^a) = s(\mu)$.
        \end{enumerate}
        We now fix~$t \in (0,1]$.
        
        \item[Step 2: Construction.] 
        {Let $\mu_t^a$ be defined as in~\eqref{eq:mu_t_severed}.}
        Note that this is an extended state that is not necessarily translation invariant. To address this, we define the translation-averaged version~$\overline{\mu}_t^a$ of~$\mu_t^a$ following the procedure from Lemma~\ref{lemma:periodization}.

        \item[Step 3: Convergence.] If $f$ is a strictly local observable, then there exists $s \in \nn$ 
        such that $f(\bq,\bp)$ only depends on $(q_i,p_i)$ with $|i| \leq s$. Moreover, if $f$ is also Lipschitz, we can let $L$ be such that 
        \begin{align*}
            |f(\bq',\bp') - f(\bq,\bp)| \leq L \max_{|i| \leq s} |(q'_i,p_i')-(q_i,p_i)|
        \end{align*}
        whenever $\max_{|i| \leq s} |(q'_i,p_i')-( q_i,p_i)| \leq 1$ (i.e.\ a Lipschitz constant). Without loss of generality, we will work with $L=1$.
        
        Note that, in view of the definition of $\overline{\mu}_t^a$ from Lemma~\ref{lemma:periodization}, since $\mu \in \mathcal{I}$ we have
        \begin{align*}
            \int f \dd\overline{\mu}_t^a 
                &= \frac{1}{\#\Lambda(a)} \sum_{j \in \Lambda(a)} \int f \dd\left(\mu_t^a \circ \mathsf{T}^*_{-j}\right) \\
                &= \frac{1}{\#\Lambda(a)} \sum_{j \in \Lambda(a)} \int \left(f \circ \mathsf{T}^*_{j} \circ \tau_t^a \circ (\mathsf{T}^*_{j})^{-1}\right) \dd\mu.
        \end{align*}
        Here, $\tau^{a,j}_t := \mathsf{T}^*_{j} \circ \tau_t^a \circ (\mathsf{T}^*_{j})^{-1}$ implements a variant of the severed dynamics~\eqref{eq:a-severed-dyn} where the center of the $a$-elementary box is translated by~$j$. Because $f$ is bounded, it suffices to show that
        \begin{equation}
        \label{eq:reduced-ent-goal}
            \lim_{a\to\infty}  \frac{1}{\#\Lambda(a)} \sum_{j \in \Lambda(a)} f \circ \tau_t^{a,j} = f \circ \tau_t
        \end{equation}
        $\mu$-almost surely. 
        To this end, fix~$(\bq,\bp) \in \Omega_2$, i.e.\ in $\mathfrak{B}_{r,C}$ for some $r \in (0,2)$ and some $C>0$\,---\,recall that $\mu(\Omega_2) = 1$ by Corollary~\ref{cor:important-I2}. It will be convenient to use the estimates from Lemma~\ref{lem:dyn-bounds} with the notation
        \begin{align*}
            \tilde{c}_t &= c_t \sup_{\gamma \in \nn} \frac{(1+s+\gamma \irange)^{r(\gamma+1)} 4\gamma^2 }{(2\gamma)!} \gamma!,
        \end{align*}
        ---\,the supremum is finite for fixed $s$, $\irange$ and $r \in (0,2)$ by e.g.\ Stirling's formula. 
        Then, 
        \begin{align*}
            &\sum_{j\in \Lambda(a)} \left|(f \circ \tau_t^{a,j})(\bq,\bp) - (f \circ \tau_t)(\bq,\bp)\right| 
            \\ 
            &\qquad \leq \sum_{j\in \Lambda(a)} \max_{|i| \leq s} \left[|q^{a,j}_i(t) - q_i(t)|+  |p^{a,j}_i(t) - p_i(t)|\right] \\
            &\qquad \leq \sum_{j\in \Lambda(a)} \max_{|i| \leq s} \inf_{\gamma \leq \gamma(i,\Lambda(a)^{\mathsf{C}}-j) }\left[\frac{\tilde{c}_t^{\gamma+1}}{\gamma!}  +\frac{\tilde{c}_t^{\gamma+1} \gamma^2}{\gamma!}\right]
            \\ 
            &\qquad \leq  2^\nu \nu\,  a^{\nu-1} \left(s \tilde{c}_t+ \sum_{k = s}^{a} \left[\frac{\tilde{c}_t^{\lceil\frac{k-s}{\irange}\rceil+1}}{\lceil\frac{k-s}{\irange}\rceil!} +\frac{\tilde{c}_t^{\lceil\frac{k-s}{\irange}\rceil+1}}{\lceil\frac{k-s}{\irange}\rceil!}\left\lceil\frac{k-s}{\irange}\right\rceil^2\right] \right)
            \\ &\qquad\leq  2^\nu \nu  \, a^{\nu-1} \left( s\tilde{c}_t + \irange \tilde c_t\, \sum_{\ell=0}^\infty \frac{\tilde c_t^\ell}{\ell!} \left[1+\ell^2 \right] \right)
            .
        \end{align*}
        We have used the fact that $\gamma(i,\Lambda(a)^{\mathsf{C}}-j) \geq \lceil\irange^{-1}(a-|j|-s)\rceil$ for $|j| \leq a-s$ (see Figure~\ref{fig:shifting-volumes}), which
        allows us to reindex the summation over~$j \in \Lambda(a)$ in terms of an index $k = a-|j|$, with each value of $k$ accounting for at most $2^\nu \nu a^{\nu-1}$ terms in the original sum over $j$.
        \begin{figure}
            \centering
            \input{fig-shifted-shells}
            \caption{A graphical justification for the inequality $\gamma(i,\Lambda(a)^{\mathsf{C}}-j) \geq \lceil\irange^{-1}(a-|j|-s)\rceil$ in cases where the dimension is $\nu = 2$, the scale of the (white) box is $a = 10$, the observable only depends on the (blue) coordinates $i$ with $\ell^\infty$-norm at most $s = 2$, and the range of the family of interactions is $D=1$. We have an inequality not only because we are considering the worst-case scenario $|i|=s$, but also due to the asymmetry in our definition of the box~$\Lambda(a)$.
            }
            \label{fig:shifting-volumes}
        \end{figure}
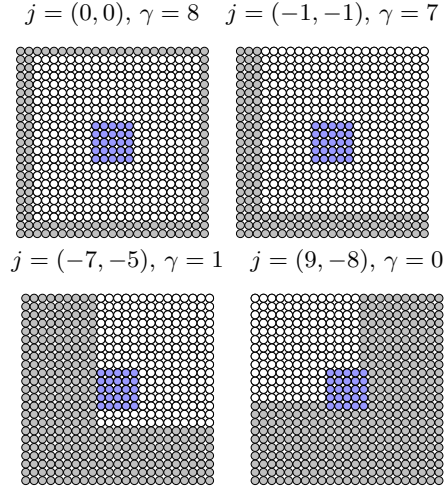
        Dividing by the volume~$\#\Lambda(a) = 2^\nu a^\nu$ and taking $a\to\infty$\,---\,keeping in mind that $s$ and $t$ (and thus $\tilde{c}_t$) are fixed\,---, we deduce that~\eqref{eq:reduced-ent-goal} indeed holds.

        \item[Step 4: Entropy.] We now fix $a \in \nn$ in addition to~$t\in(0,1]$. 
        By Liouville's theorem,
        $$
            S(\mu\res{\Lambda(na)}) = S\left(\mu_t^a\res{\Lambda(na)}\right).
        $$
        for every~$n\in\nn$. Hence, Proposition~\ref{prop:s-prop} and Lemma~\ref{lemma:periodization} prove that $s(\overline{\mu}^a_t) = s(\mu)$, as desired. \qedhere
    \end{description}
\end{proof}

\section{Dynamical equilibrium states}
\label{sec:DES}

For any $T>0$, we define the time average of the orbit of $\mu$ under the dynamics:
\begin{equation}
    \amu_T\coloneqq \frac 1 T \int_0^T\mu_t \dd t.
\end{equation}
Throughout this section we assume that (\ref{it:F}), (\ref{it:R}), (\ref{it:P1})--(\ref{it:P4}), (\ref{it:D1})--(\ref{it:D3}) and  (\ref{it:TI}) hold.

\begin{corollary}[Conservation of energy under time-averaged measures]
\label{cor:energy-conservation-averaged}
If $\mu\in\mathcal{I}_2$, then
    \begin{align*}
        \int E_0 \dd\amu_T = \int E_0 \dd\mu
    \end{align*}
    for all~$T>0$.
\end{corollary}
\begin{proof}
    The result directly follows by definition of $\amu_T$, Lemma~\ref{lem:ien-from-nien} with Fubini's theorem and Proposition~\ref{prop:energy-conservation}:
    \[
    \int_{\Omega_2}E_0\dd\amu_T=\frac 1 T\int_0^T \int_{\Omega_2}E_0\dd\mu_t \dd t=\int_{\Omega_2}E_0\dd\mu.
    \]
\end{proof}

\begin{corollary}[Conservation of entropy under time-averaged measures]
\label{cor:entropy-conservation-averaged}
    If $\mu \in \mathcal{I}_2$, then $$s(\amu_T) = s(\mu)$$ for all $T > 0$.
\end{corollary}

\begin{proof}
    Relying on Theorem~\ref{thm:conservation-entropy}, the corollary follows as in the proof of Proposition~2.16(1) in~\cite[\S{3.6}]{JPT24a}.
\end{proof}

Hereinafter, we will make the additional assumption that there exists $\epsilon \in (0,1)$ such that 
\begin{equation}
\label{eq:ss-like}
    |\min\{0,W_\Delta\}| \leq \frac{1-\epsilon}{\#\Delta} \sum_{i \in \Delta} W_{\{i\}}.
\end{equation}
for every $\Set \Subset \zz^\nu$. As with other assumptions, a constant can be added to each $W_{\{i\}}$ without any physical consequence.

\begin{lemma}
\label{lem:proto-tightness}
    Under the ongoing assumptions, if $\mu \in \mathcal{I}_2$, then the integrals 
    \begin{align*}
        \int \nien \dd\mu_t
        \qquad\text{and}\qquad
        \int \nien \dd\amu_T
    \end{align*}
    are uniformly bounded in $t\geq 0$ and $T > 0$.
\end{lemma}

\begin{proof}
    Given that $E_0\geq \nien-\sum_{\Set\supset\{0\}} |\min\{0,W_\Set\}|/{\#\Set}$, it follows from~\eqref{eq:ss-like} that there exists a constant $C$ such that 
    $
        \nien \leq C(1 + E_0).
    $
    Hence, the result follows from Proposition~\ref{prop:energy-conservation} and Corollary~\ref{cor:energy-conservation-averaged}.
\end{proof}

\begin{corollary}\label{cor:tight}
    Under the ongoing assumptions, if $\mu\in\mathcal{I}_2$, then the nets $(\mu_t)_{t \geq 0}$ and $(\amu_T)_{T> 0}$ are tight in the product topology.
\end{corollary}

\begin{proof}
    By Lemma~\ref{lem:proto-tightness}, we have integrals that are uniformly bounded in~$t \geq 0$, which we can use in Markov's inequality:
    \begin{align*}
        \mu_t(\mathfrak{B}_{r,C}^\complement)  
        &= \mu_t\{(\bq,\bp) : \nieni > C(1+|i|)^r \text{ for some } i\in \zz^\nu \} \\
        &\leq \sum_{i \in \zz^\nu} \mu_t\{(\bq,\bp) : \nieni > C(1+|i|)^r\} \\
        &\leq \sum_{i\in\zz^\nu}
            \frac{\int \nieni \dd\mu_t}{C(1+|i|)^r} \\
        &= \sum_{i\in\zz^\nu}
            \frac{\int  \nien \dd\mu_t}{C(1+|i|)^r} \\
        &= \frac{1}{C}\int \nien \dd\mu_t \sum_{i\in\zz^\nu}
            \frac{1}{(1+|i|)^r}.
    \end{align*}
    We have used that $\mu_t$ is translational invariant.
    Picking $r > \nu$, the right-hand side is finite and vanishes as $C \to \infty$, uniformly in $t \geq 0$. Since the set $\mathfrak{B}_{r,C}$
    is compact (in the product topology) for every fixed $r > 0$ and $C>0$, this yields the desired tightness property for $(\mu_t)_{t \geq 0}$. The desired tightness property for $(\amu_T)_{T > 0}$ is derived in the same way using the other type of integral in Lemma~\ref{lem:proto-tightness}.
\end{proof}

\begin{definition}[Dynamical Equilibrium States]\label{def:DEES}
    Given $\mu\in\mathcal I_2$, the weak$^*$-limit points of the net $(\amu_T)_{T> 0}$ as $T\to\infty$, denoted by $\mu_+\in \mathcal S_+(\mu)$, are called \emph{dynamical equilibrium states} (DES).\footnote{We are adding the adjective ``dynamical'' to emphasize that this is a dynamical property and not a thermodynamical property.} 
\end{definition}

The main results of this section concern the conservation of energy in finite time and the increase of entropy for DES, as well as the formulation of the approach to thermal equilibrium.

\begin{corollary}\label{cor:DESnonEmpty}
    Under the ongoing assumptions, if $\mu\in\mathcal I_2$, then the set $\mathcal S_+(\mu)$ is non-empty, and every $\mu_+ \in S_+(\mu)$ satisfies $\mu_+\circ \tau_t=\mu_+$ for all $t\in\rr$.
\end{corollary}

\begin{proof}
    Fix $\mu \in \mathcal{I}_2$. Since each $X_i$ is a complete, separable metric space, so is $\Omega$. Hence, Prokhorov's theorem applies: the tightness property in Corollary~\ref{cor:tight} implies that $(\amu_T)_{T>0}$ is sequentially compact in the weak$^*$ topology. Therefore, $\mathcal{S}_+(\mu)$ is nonempty. 

    Let $\mu_+$ be an element of $\mathcal{S}_+(\mu)$, i.e.\ a limit of some subsequence $(\amu_{T_k})_{k=1}^\infty$ with $T_k \to \infty$ as $k\to\infty$. To prove invariance under $\tau_{t_0}$, we consider a strictly local, Lipschitz observable $f$ and estimate
    \begin{align*}
        \left|\int_\Omega f \circ \tau_{t_0} \dd\mu_+ - \int_\Omega f \dd\mu_+\right|
            &= \left|\lim_{k\to\infty} \int_\Omega f \circ \tau_{t_0}  \dd\amu_{T_k} - \lim_{k\to\infty} \int_\Omega f   \dd\amu_{T_k} \right| \\
            &= \lim_{k\to\infty}\left|  \frac{1}{T_k} \int_0^{T_k} \int_\Omega f \circ \tau_{t_0} \dd\mu_{t}  \dd t - \frac{1}{T_k} \int_0^{T_k} \int_\Omega f \dd\mu_{t}  \dd t \right| \\
            &= \lim_{k\to\infty} \left| \frac{1}{T_k}\left(\int_{T_k}^{T_k + t_0} \int_\Omega f \circ \tau_{t} \dd\mu  \dd t - \int_{0}^{t_0} \int_\Omega f \circ \tau_{t} \dd\mu  \dd t \right) \right| \\
            &\leq \limsup_{k\to\infty} \frac{2|t_0| \sup f}{T_k} \\
            &= 0.
    \end{align*}
    Since $f$ was arbitrary among a measure-determining class of observables\,---\,this fact is a consequence of Lemma~\ref{lem:topo-on-states}\,---, this proves that $\mu_+ \circ \tau_{t_0}^{-1} = \mu_+$.
\end{proof}

\begin{corollary}[Conservation of energy for DES]\label{cor:consEnDES}
    Under the ongoing assumptions, if $\mu\in\mathcal I_2$ and $\mu_+\in\mathcal S_+(\mu)$, then
    \[
    \int E_0\dd\mu_+=\int E_0\dd \mu.
    \]
\end{corollary}

\begin{proof}
    Suppose $(T_k)_{k=1}^\infty$ is a sequence such that $T_k \to \infty$ and $\amu_{T_k} \to \mu_+$ weakly as $k\to\infty$.
    By the invariance property established in Corollary~\ref{cor:energy-conservation-averaged}, defining a bounded continuous truncation of $E_0$ and using weak$^*$ convergence, we obtain
    \[
    \int (g_R\circ E_0)\dd \mu_+=\lim_{k\to\infty}\int (g_R\circ E_0)\dd \amu_{T_k}=\int (g_R\circ E_0)\dd \mu
    \]
    for every $R>0$, where $g_R(x)\coloneqq\max\{-R,\min\{x,R\}\}$.
    By Lemma~\ref{lem:ien-from-nien}, $|g_R\circ E_0|\leq |E_0|$ provides an integrable dominating function, and the thesis follows by applying the dominated convergence theorem.
\end{proof}

Furthermore, Lemma~\ref{lem:VP-upper-bound} extends to DES, demonstrating that these states exhibit a thermodynamic property alongside their dynamical properties.

\begin{theorem}\label{thm:entrDES}
    Under the ongoing assumptions, if $\mu \in \mathfrak{I}_2$ and $\mu_+ \in S_+(\mu)$, then 
    \begin{align}
    \label{eq:allowed-entropies}
        s(\mu) \leq s(\mu_+) \leq \beta \int E_0 \dd\mu + p_\beta(H).
    \end{align}
    for every $\beta > 0$.
\end{theorem}

\begin{proof}
    Consider the auxiliary noninteracting energy 
    $\nifn = (\nien)^{1/\zeta}$
    for some arbitrary $\zeta > 1$.
    Thanks to~(\ref{it:P4}), we have $\Exp{-c\nifn} \in L^1(\dd\vol)$ for all $c > 0$.
    Also, by Lemma~\ref{lem:proto-tightness}, the integrals $\int |\nifn|^\zeta \dd\amu_T = \int \nien \dd\amu_T$ are uniformly bounded in $T > 0$. Hence, we can apply Proposition~\ref{prop:s-prop} with this auxiliary noninteracting energy in order to show that $s(\mu_+)$ is well-defined and that it satisfies $s(\mu_+) \geq \limsup_{T\to\infty} s(\amu_T) = s(\mu)$, thanks to the invariance property in Corollary~\ref{cor:entropy-conservation-averaged}. This establishes the desired lower bound.

    By Lemma~\ref{lem:proto-tightness} again, 
    we have $\nien \in L^1(\dd\mu_+)$. Thus, by Lemma~\ref{lem:ham-from-nien}, we have $H_\Lambda \in L^1(\dd\mu_+)$.
    For every fixed $\Lambda$, the variational upper bound obtained from Lemma~\ref{lem:rel-KL} yields
    \begin{align*}
        S_\Lambda(\mu_+\res{\Lambda}) \leq \beta \int H_\Lambda \dd\mu_+ + P_\beta(H_\Lambda).
    \end{align*}
    Dividing by the volume and taking the limit superior using Proposition~\ref{prop:s-prop}, Lemma~\ref{lem:limitHLambda} and Corollary~\ref{cor:consEnDES}, we obtain the upper bound.
\end{proof}

We now briefly address whether the upper and lower bounds in~\eqref{eq:allowed-entropies} can coincide, and if so, under what circumstances. In what follows, we call $\lambda$ a \emph{translation-invariant thermal equilibrium state at inverse temperature $\beta$}, or \emph{$\beta$-equilibrium state} for short, if it is translational invariant and it satisfies 
\begin{align*}
    s(\lambda) - \beta \int E_0 \dd\lambda = p_\beta(H).
\end{align*}

\begin{lemma}
\label{lem:existence-eq}
    Under the ongoing assumptions, for every $\beta > 0$, there exists at least one $\beta$-equilibrium state.
\end{lemma}

\begin{proof}
    Rescaling $H_\Lambda$, it suffices to consider the case $\beta = 1$. Let us construct a sequence $(\bar{\lambda}^{(a)})_{a\in\nn}$ following the procedure in Lemma~\ref{lemma:periodization} for $\lambda_{\Lambda(a)}$ defined by the relation
    \[ 
        \dd \lambda_{\Lambda(a)} = \frac{\Exp{- H_{\Lambda(a)}}}{\int \Exp{- H_{\Lambda(a)}}\dd\vol_{\Lambda(a)}} \dd\vol_{\Lambda(a)}.
    \]
    In view of~\eqref{eq:ss-like} and~(\ref{it:P4}), the normalizing constant in the denominator is indeed finite. Then, by Lemma~\ref{lemma:periodization} and a standard computation,
    \begin{align*}
        s(\bar{\lambda}^{(a)}) = \frac{S_{\Lambda(a)}(\lambda_{\Lambda(a)})}{\#\Lambda(a)} = \frac{1}{\#\Lambda(a)}\int H_{\Lambda(a)} \dd\lambda_{\Lambda(a)} + \frac{1}{\#\Lambda(a)} P(H_{\Lambda(a)}).
    \end{align*}
    Adapting the arguments used in Lemma~\ref{lem:proto-tightness} and Corollary~\ref{cor:tight}, one can show that $\bar{\lambda}^{(a)} \in \mathcal{I}_2$ for all $a\in\nn$ and that $(\bar{\lambda}^{(a)})_{a\in\nn}$ is tight.
    Then, in view of Lemma~\ref{lem:VP-upper-bound}, and adapting the arguments used for the lower bound in Theorem~\ref{thm:entrDES}, any weak limit point~$\lambda$ must lie in~$\mathcal{I}_2$ and have specific entropy equal to $\int E_0 \dd \lambda +p (H)$.
\end{proof}

\begin{remark}
\label{rem:product-equilibrium}
    It is easy to see using an argument similar to that used in the proof of Lemma~\ref{lem:S-subadd} that replacing a measure~$\lambda$ with the product of its marginal in position ($\bq$) and its marginal in momentum ($\bp$) cannot decrease the value of $s(\lambda)$ or increase the value $\beta \int E_0 \dd\lambda$. 
    In particular, if the $\beta$-equilibrium state in Lemma~\ref{lem:existence-eq} is unique, then it must be of this product form.
    
    Working with this product Ansatz, the variational principle splits into two parts as in Remark~\ref{rem:split-press}. 
    When~$X_i$ is a torus (the rotator case from Remark~\ref{rem:simplify-torus}), the variational principle for the potential part and its associated thermodynamic formalism fall within well-charted territory because of our finite-range assumption~(\ref{it:F}); see e.g.~\cite{Isr,Sim,vEFS93,Geo}. As for the momentum part, it is, although unbounded, both a simple example of a product specification and a simple example of a massive Gaussian free field; we have already noted that the pressure is strictly convex and smooth, and a straightforward computation shows that a suitable Gaussian marginal in momentum is the unique maximizer for its part of the variational principle.
\end{remark}

Given this, we say that the translation-invariant system starting in the initial state~$\mu$ has the \emph{property of approach to thermal equilibrium at inverse temperature~$\beta$} if $\mathcal{S}_+(\mu) = \{\mu_+\}$ and $\mu_+$ has an entropy that saturates the upper bound~\eqref{eq:allowed-entropies} in Theorem~\ref{thm:entrDES}, i.e.\ is a $\beta$-equilibrium state.

We only expect approach to thermal equilibrium to be conceivable if there exists a positive inverse temperature~$\beta$ that is thermodynamically compatible with the initial mean energy per particle in the sense that
\begin{align*}
    -D_\beta^+ p_\beta(H) \leq \int_\Omega E_0 \dd\mu \leq -D_\beta^- p_{\beta}(H);
\end{align*}
this in particular requires 
\begin{align*}
    -\lim_{\beta\uparrow\infty}  D_\beta^+ p_\beta(H)  \leq \int_\Omega E_0 \dd\mu \leq -\lim_{\beta\downarrow 0} D_\beta^- p_\beta(H).
\end{align*}
In view of the strict convexity property from Remark~\ref{rem:split-press},
such an inverse temperature~$\beta$ is uniquely determined by $\int E_0 \dd\mu$ in the interior of this range.

We now briefly comment on the trivial fact that ``approach to equilibrium'' could technically hold because $\mu$ is already a thermal equilibrium state.

\begin{lemma}
\label{cor:preservation-equil}
    Under the ongoing assumptions, if $\mu \in \mathcal{I}_2$ is a $\beta$-equilibrium state, then the same is true for $\mu_t$ for all~$t\in\rr$.
\end{lemma}

\begin{proof}
    Since the pressure $p_\beta$ does not depend on the evolution, this immediately follows from Proposition~\ref{prop:energy-conservation} and Theorem~\ref{thm:conservation-entropy}.
\end{proof}

\begin{corollary}
\label{cor:invar-unique-equil}
    Under the ongoing assumptions, if $\mu \in \mathcal{I}_2$ is \emph{the unique} $\beta$-equilibrium state for some $\beta$, then $\amu_T = \mu_t = \mu$ for all~$T> 0$ and~$t\in\rr$. 
\end{corollary}

This can be reformulated as a statement about the necessity, in the uniqueness regime, of a jump in the specific entropy for a ``nontrivial'' approach to equilibrium.

\begin{corollary}
    Under the ongoing assumptions, if the~$\beta$-equilibrium state is unique, then the scenario $\mu_+ = \amu_T = \mu_t = \mu$ for all~$T> 0$ and~$t\in\rr$ is the only one in which approach to equilibrium at inverse temperature $\beta$ does not require $s(\mu_+) > s(\mu)$.
\end{corollary}

Because of the lack of general, fully fleshed-out thermodynamic formalism and theory of conserved quantities, it appears reasonable to only further discuss the relations between the microscopic dynamics and macroscopic entropy growth in more specialized settings. This will be the topic of a future work, with lattices of rotators and the classical Heisenberg model as natural candidates.

\appendix 

\section{Facts on entropy}
\label{app:ent}

\subsection{Fixed finite-dimensional space}
\label{app:ent-finite}

In this subsection, we work on a fixed, finite-dimensional Euclidean space, torus or finite product thereof, equipped with its Borel $\sigma$-algebra and volume measure $\vol$. Importantly, the measure $\vol$ is not assumed to be normalizable, as this property fails for the phase space used to describe a finite region $\Lambda \Subset \zz^\nu$, i.e.\ $\bigtimes_{i \in \Lambda} T^*X_i$. 

Given a probability measure $\beta$ on this space, we set
\begin{align*}
    S(\beta) :=
        \begin{cases}
            \int -\ln \frac{\dd \beta}{\dd \vol} \dd\beta
                & \text{ if } \beta\ll\vol \text{ and the positive part of the integrand is in } L^1(\dd\beta), \\
            \infty 
                & \text{ otherwise.}
        \end{cases}
\end{align*}
As this is relative to~$\vol$, which is not necessarily normalizable, some care is required in deriving the desired properties of the map $\beta \mapsto S(\beta)$ from those of the usual relative entropy (or Kullback--Leibler divergence) for a pair of probability measures,
\begin{align*}
    S(\beta|\lambda) :=
        \begin{cases}
            \int \ln \frac{\dd \beta}{\dd \lambda} \dd\beta
                & \text{ if } \beta\ll\lambda \\
            \infty 
                & \text{ otherwise.}
        \end{cases}
\end{align*}
The following result is standard; see e.g.~\cite[\S{15.1}]{Geo}. Throughout this appendix, we will use $\phi_+$ [resp.~$\phi_-$] for the positive [resp.~negative] part of $\phi$, i.e.\ $\max\{0,\phi\}$ [resp.~$\max\{0,-\phi\}$]. In particular, $\phi = \phi_+ - \phi_-$.

\begin{lemma}
\label{lem:prop-KL}
    For every pair of probability measures $\beta$ and $\lambda$ with $\beta\ll\lambda$, the negative part of $\ln \frac{\dd\beta}{\dd\lambda}$ lies in $L^1(\dd\beta)$, and 
    $
        S(\beta|\lambda) \in [0,\infty].
    $
    Moreover, for every fixed probability measure $\lambda$, the map $\mu \mapsto S(\mu|\lambda)$ is convex and lower-semicontinuous with respect to weak convergence.
\end{lemma}

However, we ultimately want to work with $S(\beta)$ instead of $S(\beta|\lambda)$ for some auxiliary probability measure~$\lambda$ because $\vol$ is the measure that is preserved by all Hamiltonian dynamics in Liouville's theorem. While the results below are part of the folklore, we elected to provide proofs for the reader's benefit\,---\,and to dissipate any doubt about the handling of signed infinities.

\begin{lemma}
\label{lem:rel-KL}
    Suppose $\beta$ is a probability measure with $\dd\beta = \rho \dd\vol$ for some $\rho \in L^1(\dd\vol)$. If $\phi_+ \in L^1(\dd\beta)$ and $\Exp{-\phi} \in L^1(\dd\vol)$, then, with $\lambda$ the unique probability measure with $\dd\lambda \propto \Exp{-\phi} \dd\vol$, we have
    \begin{equation}
    \label{eq:rel-KL}
        S(\beta) = -S(\beta|\lambda) + \int \phi \dd\beta + \ln\int \Exp{-\phi} \dd\vol.
    \end{equation}
\end{lemma}

\begin{proof} 
    First, by Lemma~\ref{lem:prop-KL}, the negative part of $\ln \frac{\dd\beta}{\dd\lambda}$ lies in $L^1(\dd\beta)$. Second, keeping in mind that $\lambda$ is the unique probability measure with $\dd\lambda \propto \Exp{-\phi}\dd\vol$, we compute 
    \begin{align*}
        \ln \frac{\dd\beta}{\dd\lambda} 
        & 
        = \ln \frac{\dd\beta}{\dd\vol} - \ln \frac{\dd \lambda}{\dd \vol}
        = \ln \rho - \left(-\phi - \ln \int \Exp{-\phi}\dd\vol\right)
        = \ln \rho + \phi + \ln \int \Exp{-\phi}\dd\vol,
    \end{align*}
    $\beta$-almost everywhere.
    Therefore, keeping in mind that $\Exp{-\phi}\in L^1(\dd\vol)$, the assumption that $\phi_+ \in L^1(\dd\beta)$ guarantees that the negative part of $\ln \rho$ lies in $L^1(\dd\beta)$ as well.
    The above identity can be rewritten as  
    \begin{align*}
        -\ln \rho 
        &= -\ln \frac{\dd\beta}{\dd\lambda}
        + \phi + \ln\int \Exp{-\phi}\dd\vol
    \end{align*}
    where now both sides have positive parts in $L^1(\dd\beta)$. Hence, integrating with respect to $\beta$ gives the desired identity.
\end{proof}

\begin{remark}
\label{rem:var-ineq}
    If we do not assume that $\phi_+ \in L^1(\dd\beta)$, then the formula~\eqref{eq:rel-KL} may be ill defined. Indeed, the relative entropy is only known not to be~$+\infty$, so the right-hand side could involve an expression of the form $-\infty + \infty$. However, a case-by-case analysis of the possible divergences shows that, still
    \begin{align*}
        S(\beta)  \leq \int \phi \dd\beta + \ln \int \Exp{-\phi} \dd\vol
    \end{align*}
    whenever $\int \phi \dd\beta$ is well defined in $[-\infty,\infty]$.
\end{remark}

\begin{lemma}
\label{lem:lim-int-phi}
    Let $\phi$ be a measurable function, and let $\zeta > 1$ and $M > 0$ be constants. On the set $\{\beta : \int|\phi|^\zeta \dd\beta \leq M \}$, the map $\beta \mapsto \int \phi \dd\beta$ is bounded, affine and continuous with respect to weak convergence.
\end{lemma}

\begin{proof}
    This is a standard fact related to uniform integrability; see e.g.~\cite[\S{3}]{Bil}.
\end{proof}

\begin{lemma}
\label{lem:prop-S}
    Let $\phi$ be such that $\Exp{-\phi} \in L^1(\dd\vol)$, and let $\zeta > 1$ and $M > 0$ be constants. On the set $\{\beta : \int|\phi|^\zeta \dd\beta \leq M \}$, the map $\beta \mapsto S(\beta)$ is bounded above, concave and upper-semicontinuous with respect to weak convergence.
\end{lemma}

\begin{proof}
    This is follows directly from Lemmas~\ref{lem:prop-KL},~\ref{lem:rel-KL} and~\ref{lem:lim-int-phi}.
\end{proof}

\begin{lemma}
\label{lem:S-almost-convex}
    The map $\beta \mapsto S(\beta)$ is almost convex, in the sense that, whenever $S(\beta) < \infty$ and $S(\beta') < \infty$, we have $S(\varkappa \beta + (1-\varkappa)\beta') \leq \varkappa S(\beta) + (1-\varkappa)S(\beta') + \ln 2$
    for all~$\varkappa \in (0,1)$.
\end{lemma}

\begin{proof}
    Since the function $z\mapsto -\ln z$ is nonincreasing, we have, pointwise, both
    \begin{align*}
        -\ln(\varkappa\rho+(1-\varkappa)\rho') 
        &\leq -\ln\varkappa -\ln \rho &
        &\text{and}&
         -\ln(\varkappa\rho+(1-\varkappa)\rho') 
        &\leq -\ln(1-\varkappa) -\ln \rho',
    \end{align*}
    and hence
    \begin{align*}
        &-\ln(\varkappa\rho+(1-\varkappa)\rho') \, (\varkappa\rho+(1-\varkappa)\rho') \\
        &\qquad\qquad 
        \leq -\varkappa(\ln\varkappa) \rho -\varkappa(\ln \rho) \rho 
         -(1-\varkappa) (\ln(1-\varkappa)) \rho' -(1-\varkappa) (\ln \rho' ) \rho'.
    \end{align*}
    Since each term on the right-hand side has integrable positive part with respect to~$\vol$ by assumption, we can integrate both sides to obtain
    \begin{align*}
         \int -\ln(\varkappa\rho+(1-\varkappa)\rho') \dd(\varkappa\beta+(1-\varkappa)\beta')
        &\leq \varkappa  S(\beta) + (1-\varkappa) S(\beta') - [\varkappa\ln \varkappa+(1-\varkappa)\ln (1-\varkappa)].
    \end{align*}
    The result thus follows from the basic inequality $-\varkappa\ln \varkappa - (1-\varkappa)\ln (1-\varkappa) \leq \ln 2$ valid for all $\varkappa\in[0,1]$.
\end{proof}

\subsection{Thermodynamic limit}

We recall that an extended state~$\mu$ induces a consistent family~$(\mu\res{\Lambda})_{\Lambda\Subset\Omega}$ of marginals: the marginal $\mu\res{\Lambda}$ is a probability measure on $\bigtimes_{i\in\Lambda} T^*X_i$.
For a fixed finite subset~$\Lambda$ of the lattice~$\zz^\nu$, the \emph{entropy in the finite volume~$\Lambda$} is
\begin{align*}
    S_\Lambda(\mu\res{\Lambda}) \coloneqq
        \begin{cases}
            \int -\ln \frac{\dd \mu\res{\Lambda}}{\dd \vol_\Lambda} \dd\mu\res{\Lambda}
                & \text{ if } \mu\res{\Lambda} \text{is a.c.\ and the positive part of the integrand is integrable}, \\
            \infty 
                & \text{ otherwise.}
        \end{cases}
\end{align*}
For every fixed~$\Lambda \Subset \zz^\nu$, the results of Section~\ref{app:ent-finite} apply. In addition, the following lemma relates the different finite-volume entropies of a single extended state.

\begin{lemma}
\label{lem:S-subadd}
    Suppose that $\nien: T^*X_0 \to [0,\infty)$ is continuous and satisfies $\Exp{-\nien} \in L^1(\dd\vol_{\{0\}})$. If $\mu$ is a translation-invariant state such that $\nien \in L^1(\dd\mu)$, then
    the subadditivity inequality 
    \begin{align*}
        S_{\Lambda \sqcup \Lambda'}(\mu\res{\Lambda \sqcup \Lambda'}) \leq S_{\Lambda}(\mu\res{\Lambda}) + S_{\Lambda'}(\mu\res{\Lambda'})
    \end{align*}
    holds for all disjoint $\Lambda$ and $\Lambda'$. Moreover, equality holds when $\mu\res{\Lambda \sqcup \Lambda'} = \mu\res{\Lambda} \otimes \mu\res{\Lambda'}$.
\end{lemma}

\begin{proof}
    First, we note that, since $\nien$ is nonnegative and continuous, and since $\mu$ is translation invariant,
    \begin{equation}
    \label{eq:mult-exp-nien}
        \int \Exp{-\sum_{i \in \Delta} \nien \circ \mathsf{T}_i^*} \dd\vol_{\Delta} = \left( \int \Exp{-\nien } \dd\vol_{\{0\}}\right)^{\#\Delta}
    \end{equation}
    and
    \begin{equation}
    \label{eq:sum-nien}
        \int \left(\sum_{i \in \Delta} \nien \circ \mathsf{T}_i^* \right) \dd\mu\res{\Delta} = \#\Delta \ \int \nien \dd\mu\res{\{0\}}
    \end{equation}
    for every $\Delta \Subset \zz^\nu$.
    Hence, using Lemma~\ref{lem:rel-KL} with the function $\phi_\Delta = \sum_{i \in \Delta} \nien \circ \mathsf{T}_i^*$, we have $S_\Delta(\mu\res{\Delta}) < \infty$ for all $\Delta \Subset \zz^\nu$. This applies to $\Delta =  \Lambda, \Lambda', \Lambda \sqcup \Lambda'$.
            
    Note that if the negative part of $-\rho_{\Lambda \sqcup \Lambda'}\ln \rho_{\Lambda \sqcup \Lambda'}$ is not integrable, then the inequality trivially holds. Hence, we can assume that $-\rho_{\Lambda \sqcup \Lambda'}\ln\rho_{\Lambda \sqcup \Lambda'}$ is actually integrable.
    Since the negative part of $\rho_{\Lambda \sqcup \Lambda'}\ln\frac{\rho_{\Lambda \sqcup \Lambda'}}{\rho_\Lambda \rho_{\Lambda'}}$ is integrable by Lemma~\ref{lem:prop-KL}, this implies that the negative part of $-\rho_{\Lambda \sqcup \Lambda'}\ln(\rho_\Lambda\rho_{\Lambda'}) = -\rho_{\Lambda \sqcup \Lambda'}\ln\rho_\Lambda -\rho_{\Lambda \sqcup \Lambda'}\ln\rho_{\Lambda'}$ is integrable. Since the positive part of $-\rho_{\Lambda \sqcup \Lambda'}\ln \rho_{\Lambda'}$ is integrable thanks to Tonelli's theorem and our assumption, this in turn implies that the negative part of $-\rho_{\Lambda \sqcup \Lambda'}\ln \rho_\Lambda$ is integrable. Using Tonelli's theorem again, we deduce that the negative part of $-\rho_\Lambda \ln \rho_\Lambda$ is integrable. Similarly, the negative part of $-\rho_{\Lambda'} \ln \rho_{\Lambda'}$ is integrable.
    In short, it suffices to prove the statement under the assumptions that $-\rho_{\Lambda}\ln\rho_{\Lambda}$, $-\rho_{\Lambda'}\ln\rho_{\Lambda'}$, $-\rho_{\Lambda \sqcup \Lambda'}\ln\rho_{\Lambda}$, $-\rho_{\Lambda \sqcup \Lambda'}\ln\rho_{\Lambda'}$, and $-\rho_{\Lambda \sqcup \Lambda'}\ln(\rho_\Lambda\rho_{\Lambda'})$ are all integrable.
    Under this strengthened assumption, we can safely compute using Fubini's theorem:
    \begin{align*}
        S_{\Lambda\sqcup\Lambda'}(\mu\res{\Lambda\sqcup\Lambda'})
            &= \int -\ln\rho_{\Lambda \sqcup \Lambda'} \dd\mu\res{\Lambda \sqcup \Lambda'}\\
            & =  -S(\mu\res {\Lambda \sqcup \Lambda'} | \mu\res \Lambda\otimes \mu \res {\Lambda'}) \\
            &\qquad{} -\iint \rho_{\Lambda \sqcup \Lambda'}(x_\Lambda,x_{\Lambda'})\ln \rho_{\Lambda }(x_\Lambda)\dd\vol_\Lambda(x_\Lambda) \dd\vol_{\Lambda'}(x_{\Lambda'}) 
                \\ &\qquad{} -\iint \rho_{\Lambda \sqcup \Lambda'}(x_\Lambda,x_{\Lambda'})\ln \rho_{ \Lambda'}(x_{\Lambda'})\dd\vol_\Lambda(x_\Lambda) \dd\vol_{\Lambda'}(x_{\Lambda'}) \\
            & \leq -\iint \rho_{\Lambda \sqcup \Lambda'}(x_\Lambda,x_{\Lambda'})\ln \rho_{\Lambda }(x_\Lambda)\dd\vol_\Lambda(x_\Lambda) \dd\vol_{\Lambda'}(x_{\Lambda'}) 
                \\ & \qquad{} -\iint \rho_{\Lambda \sqcup \Lambda'}(x_\Lambda,x_{\Lambda'})\ln \rho_{ \Lambda'}(x_{\Lambda'})\dd\vol_\Lambda(x_\Lambda) \dd\vol_{\Lambda'}(x_{\Lambda'})\\
            & = \int- \rho_\Lambda (x_\Lambda) \ln \rho_\Lambda (x_\Lambda) \dd\vol_\Lambda(x_\Lambda) + \int -\rho_{\Lambda'} (x_{\Lambda'}) \ln \rho_{\Lambda '}(x_{\Lambda'}) \dd\vol_{\Lambda'}(x_{\Lambda'}) \\
            &= S_{\Lambda}(\mu\res{\Lambda}) + S_{\Lambda'}(\mu\res{\Lambda'}),
    \end{align*}
    as desired.
\end{proof}

\begin{proposition}
\label{prop:s-prop}
    Suppose that $\nien: T^*X_0 \to [0,\infty)$ is continuous and satisfies $\Exp{-c\nien} \in L^1(\dd\vol_{\{0\}})$ for all $c > 0$, and let $\zeta > 1$ and $M > 0$ be constants. Then, on the set $\{\mu \in \mathcal{I} : \int |\nien|^\zeta\dd\mu \leq M\}$, the limits 
    \[ 
        s(\mu) := \lim_{a\to\infty} \frac{S_{\Lambda(a)} (\mu\res{\Lambda(a)})}{\#\Lambda(a)} 
    \]
    exist in~$[-\infty,\infty)$, equal their corresponding infima, and define a map $\mu \mapsto s(\mu)$ that is bounded above, affine and upper semicontinuous with respect to weak convergence.
\end{proposition}

\begin{proof}
    Using Lemma~\ref{lem:rel-KL} and Jensen's inequality with 
    \[ 
        \phi_\Lambda = \sum_{i\in\Lambda} \nien \circ \mathsf{T}_i
    \]
    as in the proof of Lemma~\ref{lem:S-subadd}, we have that $S_{\Lambda(a)}(\mu) <\infty$ for all $a \in \nn$.
    Thanks to the subadditivity property in Lemma~\ref{lem:S-subadd}, we can then appeal to Fekete's lemma to deduce that
    \begin{align*}
        \lim_{a\to\infty}  \frac{S_{\Lambda(a)} (\mu\res{\Lambda(a)})}{\#\Lambda(a)} = \inf_{a\in\nn} \frac{S_{\Lambda(a)} (\mu\res{\Lambda(a)})}{\#\Lambda(a)},
    \end{align*}
    understood in $[-\infty,\infty)$. For every $\mu$ in the set of translation invariant measures of interest, the upper bound 
    \[ 
        S_{\{0\}}(\mu\res{\{0\}}) \leq M^{\frac 1\zeta} + \ln \int \Exp{-\nien}\dd\vol_{\{0\}}
    \]
    ---\,which is derived from Lemma~\ref{lem:rel-KL}, nonnegativity of the relative entropy and Jensen's inequality\,---\,then applies to $s(\mu)$ as well, showing that $\mu\mapsto s(\mu)$ is bounded above.
    Concavity in Lemma~\ref{lem:prop-S} and almost-convexity in Lemma~\ref{lem:S-almost-convex} show that $\mu \mapsto s(\mu)$ is affine.

    Still on the set of translation-invariant measure of interest, each functional $\mu \mapsto S_{\Lambda(a)}(\mu\res{\Lambda(a)})$ with $a \in \nn$ is upper semicontinuous with respect to weak convergence, thanks to Lemma~\ref{lem:prop-S} with 
    $$
        \psi_\Lambda = \left(\sum_{i\in\Lambda} |\nien|^{\zeta} \circ \mathsf{T}_i^* \right)^{\frac{1}{\zeta}}
    $$ 
    and the fact that weak convergence of $(\mu^{(k)})_{k\in\nn}$ implies weak convergence of $(\mu^{(k)}\res{\Lambda(a)})_{k\in\nn}$.
    To see that Lemma~\ref{lem:prop-S} indeed applies, we need to check two conditions: first,
    \begin{align*}
        \exp\left({-\psi_{\Lambda}}\right) &= \exp\left({-\left(\sum_{i\in\Lambda} |\nien|^{\zeta} \circ \mathsf{T}^*_i \right)^{\frac{1}{\zeta}}}\right) \leq \exp\left(-
             \frac{1}{\#\Lambda^{\frac{\zeta-1}{\zeta}}}
        \sum_{i \in \Lambda} |\nien| \circ \mathsf{T}^*_i \right)
    \end{align*}
    is indeed integrable with respect to $\vol_\Lambda$ by a computation similar to~\eqref{eq:mult-exp-nien} and the assumption that $\Exp{-c\nien} \in L^1(\dd\vol_{\{0\}})$ for all $c>0$; second, $\psi_\Lambda^{\zeta}$ indeed integrates to at most $M_\Lambda = \#\Lambda \, M$ with respect to $\mu\res{\Lambda}$ by a computation similar to~\eqref{eq:sum-nien} and the assumption that $\int |\nien|^\zeta \dd\mu \leq M$. Hence, upper semicontinuity of the map $\mu \mapsto s(\mu)$ follows from the fact that $s$ is a pointwise infimum of upper-semicontinuous functions.
\end{proof}

For the purpose of efficiently exploiting semicontinuity, it is convenient to introduce some terminology. 
We call bounded measurable functions on~$\Omega$ \emph{observables}. 
An observable is called \emph{strictly local} if it is $\mathcal{F}_\Lambda$-measurable for some sufficiently large~$\Lambda \Subset \zz^\nu$.
For a strictly local observable, we can speak of the Lipschitz property using e.g.~the $\ell^{\infty}$-norm on $\bigtimes_{i\in\Lambda} T^*X_i$ for some sufficiently large~$\Lambda \Subset \zz^\nu$.
The following lemma can be pieced from e.g.~\cite[\S{2.2}]{Geo},~\cite[\S{2.1}]{vEFS93} and the Portmanteau theorem.

\begin{lemma}
\label{lem:topo-on-states}
    Let $(\mu^{(k)})_{k\in\nn}$ be a sequence of extended states and $\mu$ be an extended state. Then, the following notions of weak convergence are equivalent: 
    \begin{enumerate}
        \item[(i)]  $(\mu^{(k)})_{k\in\nn}$ converges weakly $\mu$;
        \item[(ii)] for every strictly local, Lipschitz observable $f$, the sequence $\int f \dd\mu^{(k)} \to \int f \dd\mu$ as $k\to\infty$.
    \end{enumerate}
\end{lemma}
\begin{proof}
    Since strictly local, Lipschitz observables are bounded and continuous measurable functions, (i) $\Rightarrow$ (ii) is trivial.
    In order to prove (ii) $\Rightarrow$ (i), instead, we proceed to establish the nontrivial implications in the chain of equivalences:
     \begin{enumerate}
        \item[(ii)] for every strictly local, Lipschitz observable $f$, the sequence $\int f \dd\mu^{(k)} \to \int f \dd\mu$ as $k\to\infty$;
        \item[(c)] for every strictly local, continuous observable $f$, the sequence $\int f \dd\mu^{(k)} \to \int f \dd\mu$ as $k\to\infty$;
        \item[(b)] for every quasilocal, continuous observable $f$, the sequence $\int f \dd\mu^{(k)} \to \int f \dd\mu$ as $k\to\infty$;
        \item[(a)] for every uniformly continuous observable $f$, the sequence $\int f \dd\mu^{(k)} \to \int f \dd\mu$ as $k\to\infty$;
        \item[(i)]  for every continuous observable $f$, the sequence $\int f \dd\mu^{(k)} \to \int f \dd\mu$ as $k\to\infty$;      
    \end{enumerate}
    where a measurable function~$f$ on~$\Omega$ is called \emph{quasilocal} if it can be uniformly approximated by strictly local observables.
    \begin{description}
     \item[(ii) $\Rightarrow$ (c):]
    Let $f$ be a strictly local, continuous observable; then there exists a finite region $\Set \Subset \mathbb{Z}^\nu$ such that $f\in C_b(\Omega_\Set)$, with $\Omega_\Set$ the local configuration space over $\Set$ and $C_b(\Omega_\Set)$ the space of bounded continuous functions on $\Omega_\Set$. 
    For each finite $\Set \Subset \mathbb{Z}^\nu$, assuming (ii), the finite--dimensional marginals $\mu^{(k)}\res\Set$ converge weakly to $\mu\res\Set$ on $\Omega_\Set$, since bounded Lipschitz functions on $\Omega_\Set$ form a convergence-determining class on Polish spaces (Portmanteau theorem;~\cite[Thm.~11.3.3]{Dudley}).
    Then (c) immediately follows.
    \item[(c) $\Rightarrow$ (b):]
    This is a classical result (see \cite[\S2.1]{vEFS93}) based on the fact that quasilocal, continuous observables are uniform limits of strictly local, continuous observables; see~\cite[Rmk.~2.21.1]{Geo}. 
    \item[(b) $\Rightarrow$ (a):]
    Every uniformly continuous observable is, in particular, quasilocal and continuous; see~\cite[Rmk.~2.21.2]{Geo}. 
    \item[(a) $\Rightarrow$ (i):]
    This equivalence follows from the standard Portmanteau theorem~\cite[Thm.~2.1]{Bil}; indeed, uniformly continuous functions form a convergence-determining class, being dense in $C_b(\Omega_\Set)$ for the topology of uniform convergence on compact sets.
    \end{description}
    Combining the above implications, we conclude that (i) and (ii) are equivalent.
\end{proof}

Note that the existence of the entropy per unit volume is here based on translation invariance. When it comes to entropies of states defined in elementary boxes, the following lemma is useful.
    
\begin{lemma}
\label{lemma:periodization}
    Suppose that $\nien: T^*X_0  \to [0,\infty)$ is continuous and satisfies $\Exp{-c\nien} \in L^1(\dd\vol_{\{0\}})$ for all $c > 0$. 
    Let $\mu_{\Lambda(a)}$ be a probability measure on $\bigtimes_{i \in \Lambda(a)} T^*X_i$ for some $a\in\nn$, and define the measure
    \begin{align*}
        \bar\mu^{a} = \frac{1}{\#\Lambda(a)} \sum_{j \in \Lambda(a)} 
        \mu^{\otimes \zz^\nu}_{\Lambda(a)} 
        \circ \mathsf{T}^*_{-j}, 
    \end{align*}
    where $ \mu^{\otimes \zz^\nu}_{\Lambda(a)}$ denotes the periodic extension of $\mu_{\Lambda(a)}$ to the whole lattice.
    If there exists $\zeta > 1$ such that $|\nien|^\zeta \circ \mathsf{T}_i^* \in L^1(\dd\mu_{\Lambda(a)})$ for all $i\in\Lambda(a)$, then 
    \begin{equation}
        s(\bar\mu^{a}) = \frac{1}{\#\Lambda(a)}S_{\Lambda(a)}(\mu_{\Lambda(a)}).
    \end{equation}
\end{lemma}

\begin{proof}
    For the upper bound, the restricted measure $\mu^{\otimes \zz^\nu}_{\Lambda(a)} \circ \mathsf{T}^*_{-j}\res{\Lambda(na)}$ can be decomposed as a product over $(n-1)^\nu$ disjoint sub-cubes of size $\Lambda(a)$ that fit entirely inside $\Lambda(na)$, together with a boundary layer of volume $O(n^{\nu-1})$. For the lower bound, we consider a disjoint union of $(n+1)^\nu$ cubes of size $\Lambda(a)$ that cover (i.e.\ entirely contain) $\Lambda(na)$, with a leftover region consisting of $O(n^{\nu-1})$ boundary points.
    Hence, by subadditivity (Lemma~\ref{lem:S-subadd}) and invariance under translations by vectors in $2a\zz^\nu$, we have
    \begin{align*}
        &(n+1)^\nu  S_{\Lambda(a)}\left( \mu_{\Lambda(a)}  \right) - O(n^{\nu-1}) \max_{i\in\Lambda(a)}\max\left\{0,S_{\{i\}}(\mu_{\Lambda(a)}\res{\{i\}})\right\} \\
        &\qquad\qquad \leq 
        S\left( \mu^{\otimes \zz^\nu}_{\Lambda(a)} 
        \circ \mathsf{T}^*_{-j}
        \res{\Lambda(na)} \right) \\
        &\qquad\qquad \leq (n-1)^\nu S_{\Lambda(a)}\left( \mu_{\Lambda(a)} \right) + 
        O(n^{\nu-1}) \max_{i\in\Lambda(a)}\max\left\{0,S_{\{i\}}(\mu_{\Lambda(a)}\res{\{i\}})\right\}
    \end{align*}
    for each $j \in \Lambda(a)$ and $n\in\nn$. 
    
    By Jensen's inequality and Lemma~\ref{lem:rel-KL} with $\phi_{\{i\}} = \nien \circ \mathsf{T}_i^*$, the maximum appearing in these inequalities is finite. 
    So, using concavity for the lower bound and almost convexity for the upper bound (Lemmas~\ref{lem:prop-S} and~\ref{lem:S-almost-convex}, respectively), we have
    \begin{align*}
        (n+1)^\nu  S_{\Lambda(a)}\left( \mu_{\Lambda(a)}  \right) - 
        O(n^{\nu-1})
        &\leq 
        S_{\Lambda(na)}\left( \bar\mu^{a}\res{\Lambda(na)} \right) \\
        &\leq (n-1)^\nu S_{\Lambda(a)}\left( \mu_{\Lambda(a)} \right) + 
        O(n^{\nu-1}) + O(\ln n)
    \end{align*}
    for every $n\in\nn$. Dividing by $\#\Lambda(na)=n^\nu \#\Lambda(a)$, we find 
     \begin{align*}
        \frac{(n+1)^\nu}{n^\nu \#\Lambda(a)}  S_{\Lambda(a)}\left( \mu_{\Lambda(a)}  \right) - O(n^{-1})
        &\leq 
        \frac{1}{\#\Lambda(na)}S_{\Lambda(na)}\left( \bar\mu^{a}\res{\Lambda(na)} \right) \\
        &\leq \frac{(n-1)^\nu}{n^\nu \#\Lambda(a)} S_{\Lambda(a)}\left( \mu_{\Lambda(a)} \right) +       O(n^{-1}) + O(n^{-\nu}\ln n).
    \end{align*}
    Since there exists $\zeta > 1$ such that $|\nien|^\zeta \circ \mathsf{T}_i^* \in L^1(\dd\mu_{\Lambda(a)})$ for all $i\in\Lambda(a)$, and since $\#\Lambda(a)$ is finite, we can use Proposition~\ref{prop:s-prop} with 
    \[
        \overline{\nien} = \frac{1}{\#\Lambda(a)} \sum_{i \in \Lambda(a)} \nien \circ \mathsf{T}_i^*
    \]
    to take the limit $n \to \infty$ and derive the desired identity.
\end{proof}

\begin{remark}
\label{rem:concrete-frak-F}
    In the main body of the article, it is the fact that $\mu \in \mathcal{I}_2$, defined in terms of one-body potentials satisfying~(\ref{it:TI}) and~(\ref{it:P4}),
    that ensures the applicability of the above results.
\end{remark}

\end{document}

%% file: fig-shifted-shells.tex
\begin{tikzpicture}
\begin{axis}[
    width=4cm,
    height=4cm,
    xmin=-11, xmax=11,
    ymin=-11, ymax=11,
    axis equal,
    title={$j = (0,0)$, $\gamma = 8$},
    title style={font=\small},
    xlabel={}, ylabel={},
    xtick=\empty, ytick=\empty,
]

\addplot [only marks, mark=*, mark size=1.5pt, fill=gray!50, draw=black, line width=0.2pt]
table [row sep=\\] {
x y \\
-11 -11 \\
-11 -10 \\
-11 -9 \\
-11 -8 \\
-11 -7 \\
-11 -6 \\
-11 -5 \\
-11 -4 \\
-11 -3 \\
-11 -2 \\
-11 -1 \\
-11 0 \\
-11 1 \\
-11 2 \\
-11 3 \\
-11 4 \\
-11 5 \\
-11 6 \\
-11 7 \\
-11 8 \\
-11 9 \\
-11 10 \\
-11 11 \\
-10 -11 \\
-10 -10 \\
-10 -9 \\
-10 -8 \\
-10 -7 \\
-10 -6 \\
-10 -5 \\
-10 -4 \\
-10 -3 \\
-10 -2 \\
-10 -1 \\
-10 0 \\
-10 1 \\
-10 2 \\
-10 3 \\
-10 4 \\
-10 5 \\
-10 6 \\
-10 7 \\
-10 8 \\
-10 9 \\
-10 10 \\
-10 11 \\
-9 -11 \\
-9 -10 \\
-9 -9 \\
-9 -8 \\
-9 -7 \\
-9 -6 \\
-9 -5 \\
-9 -4 \\
-9 -3 \\
-9 -2 \\
-9 -1 \\
-9 0 \\
-9 1 \\
-9 2 \\
-9 3 \\
-9 4 \\
-9 5 \\
-9 6 \\
-9 7 \\
-9 8 \\
-9 9 \\
-9 10 \\
-9 11 \\
-8 -11 \\
-8 -10 \\
-8 -9 \\
-8 -8 \\
-8 -7 \\
-8 -6 \\
-8 -5 \\
-8 -4 \\
-8 -3 \\
-8 -2 \\
-8 -1 \\
-8 0 \\
-8 1 \\
-8 2 \\
-8 3 \\
-8 4 \\
-8 5 \\
-8 6 \\
-8 7 \\
-8 8 \\
-8 9 \\
-8 10 \\
-8 11 \\
-7 -11 \\
-7 -10 \\
-7 -9 \\
-7 -8 \\
-7 -7 \\
-7 -6 \\
-7 -5 \\
-7 -4 \\
-7 -3 \\
-7 -2 \\
-7 -1 \\
-7 0 \\
-7 1 \\
-7 2 \\
-7 3 \\
-7 4 \\
-7 5 \\
-7 6 \\
-7 7 \\
-7 8 \\
-7 9 \\
-7 10 \\
-7 11 \\
-6 -11 \\
-6 -10 \\
-6 -9 \\
-6 -8 \\
-6 -7 \\
-6 -6 \\
-6 -5 \\
-6 -4 \\
-6 -3 \\
-6 -2 \\
-6 -1 \\
-6 0 \\
-6 1 \\
-6 2 \\
-6 3 \\
-6 4 \\
-6 5 \\
-6 6 \\
-6 7 \\
-6 8 \\
-6 9 \\
-6 10 \\
-6 11 \\
-5 -11 \\
-5 -10 \\
-5 -9 \\
-5 -8 \\
-5 -7 \\
-5 -6 \\
-5 -5 \\
-5 -4 \\
-5 -3 \\
-5 -2 \\
-5 -1 \\
-5 0 \\
-5 1 \\
-5 2 \\
-5 3 \\
-5 4 \\
-5 5 \\
-5 6 \\
-5 7 \\
-5 8 \\
-5 9 \\
-5 10 \\
-5 11 \\
-4 -11 \\
-4 -10 \\
-4 -9 \\
-4 -8 \\
-4 -7 \\
-4 -6 \\
-4 -5 \\
-4 -4 \\
-4 -3 \\
-4 -2 \\
-4 -1 \\
-4 0 \\
-4 1 \\
-4 2 \\
-4 3 \\
-4 4 \\
-4 5 \\
-4 6 \\
-4 7 \\
-4 8 \\
-4 9 \\
-4 10 \\
-4 11 \\
-3 -11 \\
-3 -10 \\
-3 -9 \\
-3 -8 \\
-3 -7 \\
-3 -6 \\
-3 -5 \\
-3 -4 \\
-3 -3 \\
-3 -2 \\
-3 -1 \\
-3 0 \\
-3 1 \\
-3 2 \\
-3 3 \\
-3 4 \\
-3 5 \\
-3 6 \\
-3 7 \\
-3 8 \\
-3 9 \\
-3 10 \\
-3 11 \\
-2 -11 \\
-2 -10 \\
-2 -9 \\
-2 -8 \\
-2 -7 \\
-2 -6 \\
-2 -5 \\
-2 -4 \\
-2 -3 \\
-2 -2 \\
-2 -1 \\
-2 0 \\
-2 1 \\
-2 2 \\
-2 3 \\
-2 4 \\
-2 5 \\
-2 6 \\
-2 7 \\
-2 8 \\
-2 9 \\
-2 10 \\
-2 11 \\
-1 -11 \\
-1 -10 \\
-1 -9 \\
-1 -8 \\
-1 -7 \\
-1 -6 \\
-1 -5 \\
-1 -4 \\
-1 -3 \\
-1 -2 \\
-1 -1 \\
-1 0 \\
-1 1 \\
-1 2 \\
-1 3 \\
-1 4 \\
-1 5 \\
-1 6 \\
-1 7 \\
-1 8 \\
-1 9 \\
-1 10 \\
-1 11 \\
0 -11 \\
0 -10 \\
0 -9 \\
0 -8 \\
0 -7 \\
0 -6 \\
0 -5 \\
0 -4 \\
0 -3 \\
0 -2 \\
0 -1 \\
0 0 \\
0 1 \\
0 2 \\
0 3 \\
0 4 \\
0 5 \\
0 6 \\
0 7 \\
0 8 \\
0 9 \\
0 10 \\
0 11 \\
1 -11 \\
1 -10 \\
1 -9 \\
1 -8 \\
1 -7 \\
1 -6 \\
1 -5 \\
1 -4 \\
1 -3 \\
1 -2 \\
1 -1 \\
1 0 \\
1 1 \\
1 2 \\
1 3 \\
1 4 \\
1 5 \\
1 6 \\
1 7 \\
1 8 \\
1 9 \\
1 10 \\
1 11 \\
2 -11 \\
2 -10 \\
2 -9 \\
2 -8 \\
2 -7 \\
2 -6 \\
2 -5 \\
2 -4 \\
2 -3 \\
2 -2 \\
2 -1 \\
2 0 \\
2 1 \\
2 2 \\
2 3 \\
2 4 \\
2 5 \\
2 6 \\
2 7 \\
2 8 \\
2 9 \\
2 10 \\
2 11 \\
3 -11 \\
3 -10 \\
3 -9 \\
3 -8 \\
3 -7 \\
3 -6 \\
3 -5 \\
3 -4 \\
3 -3 \\
3 -2 \\
3 -1 \\
3 0 \\
3 1 \\
3 2 \\
3 3 \\
3 4 \\
3 5 \\
3 6 \\
3 7 \\
3 8 \\
3 9 \\
3 10 \\
3 11 \\
4 -11 \\
4 -10 \\
4 -9 \\
4 -8 \\
4 -7 \\
4 -6 \\
4 -5 \\
4 -4 \\
4 -3 \\
4 -2 \\
4 -1 \\
4 0 \\
4 1 \\
4 2 \\
4 3 \\
4 4 \\
4 5 \\
4 6 \\
4 7 \\
4 8 \\
4 9 \\
4 10 \\
4 11 \\
5 -11 \\
5 -10 \\
5 -9 \\
5 -8 \\
5 -7 \\
5 -6 \\
5 -5 \\
5 -4 \\
5 -3 \\
5 -2 \\
5 -1 \\
5 0 \\
5 1 \\
5 2 \\
5 3 \\
5 4 \\
5 5 \\
5 6 \\
5 7 \\
5 8 \\
5 9 \\
5 10 \\
5 11 \\
6 -11 \\
6 -10 \\
6 -9 \\
6 -8 \\
6 -7 \\
6 -6 \\
6 -5 \\
6 -4 \\
6 -3 \\
6 -2 \\
6 -1 \\
6 0 \\
6 1 \\
6 2 \\
6 3 \\
6 4 \\
6 5 \\
6 6 \\
6 7 \\
6 8 \\
6 9 \\
6 10 \\
6 11 \\
7 -11 \\
7 -10 \\
7 -9 \\
7 -8 \\
7 -7 \\
7 -6 \\
7 -5 \\
7 -4 \\
7 -3 \\
7 -2 \\
7 -1 \\
7 0 \\
7 1 \\
7 2 \\
7 3 \\
7 4 \\
7 5 \\
7 6 \\
7 7 \\
7 8 \\
7 9 \\
7 10 \\
7 11 \\
8 -11 \\
8 -10 \\
8 -9 \\
8 -8 \\
8 -7 \\
8 -6 \\
8 -5 \\
8 -4 \\
8 -3 \\
8 -2 \\
8 -1 \\
8 0 \\
8 1 \\
8 2 \\
8 3 \\
8 4 \\
8 5 \\
8 6 \\
8 7 \\
8 8 \\
8 9 \\
8 10 \\
8 11 \\
9 -11 \\
9 -10 \\
9 -9 \\
9 -8 \\
9 -7 \\
9 -6 \\
9 -5 \\
9 -4 \\
9 -3 \\
9 -2 \\
9 -1 \\
9 0 \\
9 1 \\
9 2 \\
9 3 \\
9 4 \\
9 5 \\
9 6 \\
9 7 \\
9 8 \\
9 9 \\
9 10 \\
9 11 \\
10 -11 \\
10 -10 \\
10 -9 \\
10 -8 \\
10 -7 \\
10 -6 \\
10 -5 \\
10 -4 \\
10 -3 \\
10 -2 \\
10 -1 \\
10 0 \\
10 1 \\
10 2 \\
10 3 \\
10 4 \\
10 5 \\
10 6 \\
10 7 \\
10 8 \\
10 9 \\
10 10 \\
10 11 \\
11 -11 \\
11 -10 \\
11 -9 \\
11 -8 \\
11 -7 \\
11 -6 \\
11 -5 \\
11 -4 \\
11 -3 \\
11 -2 \\
11 -1 \\
11 0 \\
11 1 \\
11 2 \\
11 3 \\
11 4 \\
11 5 \\
11 6 \\
11 7 \\
11 8 \\
11 9 \\
11 10 \\
11 11 \\
};

\addplot [only marks, mark=*, mark size=1.5pt, fill=white, draw=black, line width=0.2pt]
table [row sep=\\] {
x y \\
-9 -9 \\
-9 -8 \\
-9 -7 \\
-9 -6 \\
-9 -5 \\
-9 -4 \\
-9 -3 \\
-9 -2 \\
-9 -1 \\
-9 0 \\
-9 1 \\
-9 2 \\
-9 3 \\
-9 4 \\
-9 5 \\
-9 6 \\
-9 7 \\
-9 8 \\
-9 9 \\
-9 10 \\
-8 -9 \\
-8 -8 \\
-8 -7 \\
-8 -6 \\
-8 -5 \\
-8 -4 \\
-8 -3 \\
-8 -2 \\
-8 -1 \\
-8 0 \\
-8 1 \\
-8 2 \\
-8 3 \\
-8 4 \\
-8 5 \\
-8 6 \\
-8 7 \\
-8 8 \\
-8 9 \\
-8 10 \\
-7 -9 \\
-7 -8 \\
-7 -7 \\
-7 -6 \\
-7 -5 \\
-7 -4 \\
-7 -3 \\
-7 -2 \\
-7 -1 \\
-7 0 \\
-7 1 \\
-7 2 \\
-7 3 \\
-7 4 \\
-7 5 \\
-7 6 \\
-7 7 \\
-7 8 \\
-7 9 \\
-7 10 \\
-6 -9 \\
-6 -8 \\
-6 -7 \\
-6 -6 \\
-6 -5 \\
-6 -4 \\
-6 -3 \\
-6 -2 \\
-6 -1 \\
-6 0 \\
-6 1 \\
-6 2 \\
-6 3 \\
-6 4 \\
-6 5 \\
-6 6 \\
-6 7 \\
-6 8 \\
-6 9 \\
-6 10 \\
-5 -9 \\
-5 -8 \\
-5 -7 \\
-5 -6 \\
-5 -5 \\
-5 -4 \\
-5 -3 \\
-5 -2 \\
-5 -1 \\
-5 0 \\
-5 1 \\
-5 2 \\
-5 3 \\
-5 4 \\
-5 5 \\
-5 6 \\
-5 7 \\
-5 8 \\
-5 9 \\
-5 10 \\
-4 -9 \\
-4 -8 \\
-4 -7 \\
-4 -6 \\
-4 -5 \\
-4 -4 \\
-4 -3 \\
-4 -2 \\
-4 -1 \\
-4 0 \\
-4 1 \\
-4 2 \\
-4 3 \\
-4 4 \\
-4 5 \\
-4 6 \\
-4 7 \\
-4 8 \\
-4 9 \\
-4 10 \\
-3 -9 \\
-3 -8 \\
-3 -7 \\
-3 -6 \\
-3 -5 \\
-3 -4 \\
-3 -3 \\
-3 -2 \\
-3 -1 \\
-3 0 \\
-3 1 \\
-3 2 \\
-3 3 \\
-3 4 \\
-3 5 \\
-3 6 \\
-3 7 \\
-3 8 \\
-3 9 \\
-3 10 \\
-2 -9 \\
-2 -8 \\
-2 -7 \\
-2 -6 \\
-2 -5 \\
-2 -4 \\
-2 -3 \\
-2 -2 \\
-2 -1 \\
-2 0 \\
-2 1 \\
-2 2 \\
-2 3 \\
-2 4 \\
-2 5 \\
-2 6 \\
-2 7 \\
-2 8 \\
-2 9 \\
-2 10 \\
-1 -9 \\
-1 -8 \\
-1 -7 \\
-1 -6 \\
-1 -5 \\
-1 -4 \\
-1 -3 \\
-1 -2 \\
-1 -1 \\
-1 0 \\
-1 1 \\
-1 2 \\
-1 3 \\
-1 4 \\
-1 5 \\
-1 6 \\
-1 7 \\
-1 8 \\
-1 9 \\
-1 10 \\
0 -9 \\
0 -8 \\
0 -7 \\
0 -6 \\
0 -5 \\
0 -4 \\
0 -3 \\
0 -2 \\
0 -1 \\
0 0 \\
0 1 \\
0 2 \\
0 3 \\
0 4 \\
0 5 \\
0 6 \\
0 7 \\
0 8 \\
0 9 \\
0 10 \\
1 -9 \\
1 -8 \\
1 -7 \\
1 -6 \\
1 -5 \\
1 -4 \\
1 -3 \\
1 -2 \\
1 -1 \\
1 0 \\
1 1 \\
1 2 \\
1 3 \\
1 4 \\
1 5 \\
1 6 \\
1 7 \\
1 8 \\
1 9 \\
1 10 \\
2 -9 \\
2 -8 \\
2 -7 \\
2 -6 \\
2 -5 \\
2 -4 \\
2 -3 \\
2 -2 \\
2 -1 \\
2 0 \\
2 1 \\
2 2 \\
2 3 \\
2 4 \\
2 5 \\
2 6 \\
2 7 \\
2 8 \\
2 9 \\
2 10 \\
3 -9 \\
3 -8 \\
3 -7 \\
3 -6 \\
3 -5 \\
3 -4 \\
3 -3 \\
3 -2 \\
3 -1 \\
3 0 \\
3 1 \\
3 2 \\
3 3 \\
3 4 \\
3 5 \\
3 6 \\
3 7 \\
3 8 \\
3 9 \\
3 10 \\
4 -9 \\
4 -8 \\
4 -7 \\
4 -6 \\
4 -5 \\
4 -4 \\
4 -3 \\
4 -2 \\
4 -1 \\
4 0 \\
4 1 \\
4 2 \\
4 3 \\
4 4 \\
4 5 \\
4 6 \\
4 7 \\
4 8 \\
4 9 \\
4 10 \\
5 -9 \\
5 -8 \\
5 -7 \\
5 -6 \\
5 -5 \\
5 -4 \\
5 -3 \\
5 -2 \\
5 -1 \\
5 0 \\
5 1 \\
5 2 \\
5 3 \\
5 4 \\
5 5 \\
5 6 \\
5 7 \\
5 8 \\
5 9 \\
5 10 \\
6 -9 \\
6 -8 \\
6 -7 \\
6 -6 \\
6 -5 \\
6 -4 \\
6 -3 \\
6 -2 \\
6 -1 \\
6 0 \\
6 1 \\
6 2 \\
6 3 \\
6 4 \\
6 5 \\
6 6 \\
6 7 \\
6 8 \\
6 9 \\
6 10 \\
7 -9 \\
7 -8 \\
7 -7 \\
7 -6 \\
7 -5 \\
7 -4 \\
7 -3 \\
7 -2 \\
7 -1 \\
7 0 \\
7 1 \\
7 2 \\
7 3 \\
7 4 \\
7 5 \\
7 6 \\
7 7 \\
7 8 \\
7 9 \\
7 10 \\
8 -9 \\
8 -8 \\
8 -7 \\
8 -6 \\
8 -5 \\
8 -4 \\
8 -3 \\
8 -2 \\
8 -1 \\
8 0 \\
8 1 \\
8 2 \\
8 3 \\
8 4 \\
8 5 \\
8 6 \\
8 7 \\
8 8 \\
8 9 \\
8 10 \\
9 -9 \\
9 -8 \\
9 -7 \\
9 -6 \\
9 -5 \\
9 -4 \\
9 -3 \\
9 -2 \\
9 -1 \\
9 0 \\
9 1 \\
9 2 \\
9 3 \\
9 4 \\
9 5 \\
9 6 \\
9 7 \\
9 8 \\
9 9 \\
9 10 \\
10 -9 \\
10 -8 \\
10 -7 \\
10 -6 \\
10 -5 \\
10 -4 \\
10 -3 \\
10 -2 \\
10 -1 \\
10 0 \\
10 1 \\
10 2 \\
10 3 \\
10 4 \\
10 5 \\
10 6 \\
10 7 \\
10 8 \\
10 9 \\
10 10 \\
};

\addplot [only marks, mark=*, mark size=1.5pt, fill=blue!40, draw=black, line width=0.2pt]
table [row sep=\\] {
x y \\
-2 -2 \\
-2 -1 \\
-2 0 \\
-2 1 \\
-2 2 \\
-1 -2 \\
-1 -1 \\
-1 0 \\
-1 1 \\
-1 2 \\
0 -2 \\
0 -1 \\
0 0 \\
0 1 \\
0 2 \\
1 -2 \\
1 -1 \\
1 0 \\
1 1 \\
1 2 \\
2 -2 \\
2 -1 \\
2 0 \\
2 1 \\
2 2 \\
};
\end{axis}
\end{tikzpicture}
\begin{tikzpicture}
\begin{axis}[
    width=4cm,
    height=4cm,
    xmin=-11, xmax=11,
    ymin=-11, ymax=11,
    axis equal,
    title={$j = (-1,-1)$, $\gamma = 7$},
    title style={font=\small},
    xlabel={}, ylabel={},
    xtick=\empty, ytick=\empty,
]

\addplot [only marks, mark=*, mark size=1.5pt, fill=gray!50, draw=black, line width=0.2pt]
table [row sep=\\] {
x y \\
-11 -11 \\
-11 -10 \\
-11 -9 \\
-11 -8 \\
-11 -7 \\
-11 -6 \\
-11 -5 \\
-11 -4 \\
-11 -3 \\
-11 -2 \\
-11 -1 \\
-11 0 \\
-11 1 \\
-11 2 \\
-11 3 \\
-11 4 \\
-11 5 \\
-11 6 \\
-11 7 \\
-11 8 \\
-11 9 \\
-11 10 \\
-11 11 \\
-10 -11 \\
-10 -10 \\
-10 -9 \\
-10 -8 \\
-10 -7 \\
-10 -6 \\
-10 -5 \\
-10 -4 \\
-10 -3 \\
-10 -2 \\
-10 -1 \\
-10 0 \\
-10 1 \\
-10 2 \\
-10 3 \\
-10 4 \\
-10 5 \\
-10 6 \\
-10 7 \\
-10 8 \\
-10 9 \\
-10 10 \\
-10 11 \\
-9 -11 \\
-9 -10 \\
-9 -9 \\
-9 -8 \\
-9 -7 \\
-9 -6 \\
-9 -5 \\
-9 -4 \\
-9 -3 \\
-9 -2 \\
-9 -1 \\
-9 0 \\
-9 1 \\
-9 2 \\
-9 3 \\
-9 4 \\
-9 5 \\
-9 6 \\
-9 7 \\
-9 8 \\
-9 9 \\
-9 10 \\
-9 11 \\
-8 -11 \\
-8 -10 \\
-8 -9 \\
-8 -8 \\
-8 -7 \\
-8 -6 \\
-8 -5 \\
-8 -4 \\
-8 -3 \\
-8 -2 \\
-8 -1 \\
-8 0 \\
-8 1 \\
-8 2 \\
-8 3 \\
-8 4 \\
-8 5 \\
-8 6 \\
-8 7 \\
-8 8 \\
-8 9 \\
-8 10 \\
-8 11 \\
-7 -11 \\
-7 -10 \\
-7 -9 \\
-7 -8 \\
-7 -7 \\
-7 -6 \\
-7 -5 \\
-7 -4 \\
-7 -3 \\
-7 -2 \\
-7 -1 \\
-7 0 \\
-7 1 \\
-7 2 \\
-7 3 \\
-7 4 \\
-7 5 \\
-7 6 \\
-7 7 \\
-7 8 \\
-7 9 \\
-7 10 \\
-7 11 \\
-6 -11 \\
-6 -10 \\
-6 -9 \\
-6 -8 \\
-6 -7 \\
-6 -6 \\
-6 -5 \\
-6 -4 \\
-6 -3 \\
-6 -2 \\
-6 -1 \\
-6 0 \\
-6 1 \\
-6 2 \\
-6 3 \\
-6 4 \\
-6 5 \\
-6 6 \\
-6 7 \\
-6 8 \\
-6 9 \\
-6 10 \\
-6 11 \\
-5 -11 \\
-5 -10 \\
-5 -9 \\
-5 -8 \\
-5 -7 \\
-5 -6 \\
-5 -5 \\
-5 -4 \\
-5 -3 \\
-5 -2 \\
-5 -1 \\
-5 0 \\
-5 1 \\
-5 2 \\
-5 3 \\
-5 4 \\
-5 5 \\
-5 6 \\
-5 7 \\
-5 8 \\
-5 9 \\
-5 10 \\
-5 11 \\
-4 -11 \\
-4 -10 \\
-4 -9 \\
-4 -8 \\
-4 -7 \\
-4 -6 \\
-4 -5 \\
-4 -4 \\
-4 -3 \\
-4 -2 \\
-4 -1 \\
-4 0 \\
-4 1 \\
-4 2 \\
-4 3 \\
-4 4 \\
-4 5 \\
-4 6 \\
-4 7 \\
-4 8 \\
-4 9 \\
-4 10 \\
-4 11 \\
-3 -11 \\
-3 -10 \\
-3 -9 \\
-3 -8 \\
-3 -7 \\
-3 -6 \\
-3 -5 \\
-3 -4 \\
-3 -3 \\
-3 -2 \\
-3 -1 \\
-3 0 \\
-3 1 \\
-3 2 \\
-3 3 \\
-3 4 \\
-3 5 \\
-3 6 \\
-3 7 \\
-3 8 \\
-3 9 \\
-3 10 \\
-3 11 \\
-2 -11 \\
-2 -10 \\
-2 -9 \\
-2 -8 \\
-2 -7 \\
-2 -6 \\
-2 -5 \\
-2 -4 \\
-2 -3 \\
-2 -2 \\
-2 -1 \\
-2 0 \\
-2 1 \\
-2 2 \\
-2 3 \\
-2 4 \\
-2 5 \\
-2 6 \\
-2 7 \\
-2 8 \\
-2 9 \\
-2 10 \\
-2 11 \\
-1 -11 \\
-1 -10 \\
-1 -9 \\
-1 -8 \\
-1 -7 \\
-1 -6 \\
-1 -5 \\
-1 -4 \\
-1 -3 \\
-1 -2 \\
-1 -1 \\
-1 0 \\
-1 1 \\
-1 2 \\
-1 3 \\
-1 4 \\
-1 5 \\
-1 6 \\
-1 7 \\
-1 8 \\
-1 9 \\
-1 10 \\
-1 11 \\
0 -11 \\
0 -10 \\
0 -9 \\
0 -8 \\
0 -7 \\
0 -6 \\
0 -5 \\
0 -4 \\
0 -3 \\
0 -2 \\
0 -1 \\
0 0 \\
0 1 \\
0 2 \\
0 3 \\
0 4 \\
0 5 \\
0 6 \\
0 7 \\
0 8 \\
0 9 \\
0 10 \\
0 11 \\
1 -11 \\
1 -10 \\
1 -9 \\
1 -8 \\
1 -7 \\
1 -6 \\
1 -5 \\
1 -4 \\
1 -3 \\
1 -2 \\
1 -1 \\
1 0 \\
1 1 \\
1 2 \\
1 3 \\
1 4 \\
1 5 \\
1 6 \\
1 7 \\
1 8 \\
1 9 \\
1 10 \\
1 11 \\
2 -11 \\
2 -10 \\
2 -9 \\
2 -8 \\
2 -7 \\
2 -6 \\
2 -5 \\
2 -4 \\
2 -3 \\
2 -2 \\
2 -1 \\
2 0 \\
2 1 \\
2 2 \\
2 3 \\
2 4 \\
2 5 \\
2 6 \\
2 7 \\
2 8 \\
2 9 \\
2 10 \\
2 11 \\
3 -11 \\
3 -10 \\
3 -9 \\
3 -8 \\
3 -7 \\
3 -6 \\
3 -5 \\
3 -4 \\
3 -3 \\
3 -2 \\
3 -1 \\
3 0 \\
3 1 \\
3 2 \\
3 3 \\
3 4 \\
3 5 \\
3 6 \\
3 7 \\
3 8 \\
3 9 \\
3 10 \\
3 11 \\
4 -11 \\
4 -10 \\
4 -9 \\
4 -8 \\
4 -7 \\
4 -6 \\
4 -5 \\
4 -4 \\
4 -3 \\
4 -2 \\
4 -1 \\
4 0 \\
4 1 \\
4 2 \\
4 3 \\
4 4 \\
4 5 \\
4 6 \\
4 7 \\
4 8 \\
4 9 \\
4 10 \\
4 11 \\
5 -11 \\
5 -10 \\
5 -9 \\
5 -8 \\
5 -7 \\
5 -6 \\
5 -5 \\
5 -4 \\
5 -3 \\
5 -2 \\
5 -1 \\
5 0 \\
5 1 \\
5 2 \\
5 3 \\
5 4 \\
5 5 \\
5 6 \\
5 7 \\
5 8 \\
5 9 \\
5 10 \\
5 11 \\
6 -11 \\
6 -10 \\
6 -9 \\
6 -8 \\
6 -7 \\
6 -6 \\
6 -5 \\
6 -4 \\
6 -3 \\
6 -2 \\
6 -1 \\
6 0 \\
6 1 \\
6 2 \\
6 3 \\
6 4 \\
6 5 \\
6 6 \\
6 7 \\
6 8 \\
6 9 \\
6 10 \\
6 11 \\
7 -11 \\
7 -10 \\
7 -9 \\
7 -8 \\
7 -7 \\
7 -6 \\
7 -5 \\
7 -4 \\
7 -3 \\
7 -2 \\
7 -1 \\
7 0 \\
7 1 \\
7 2 \\
7 3 \\
7 4 \\
7 5 \\
7 6 \\
7 7 \\
7 8 \\
7 9 \\
7 10 \\
7 11 \\
8 -11 \\
8 -10 \\
8 -9 \\
8 -8 \\
8 -7 \\
8 -6 \\
8 -5 \\
8 -4 \\
8 -3 \\
8 -2 \\
8 -1 \\
8 0 \\
8 1 \\
8 2 \\
8 3 \\
8 4 \\
8 5 \\
8 6 \\
8 7 \\
8 8 \\
8 9 \\
8 10 \\
8 11 \\
9 -11 \\
9 -10 \\
9 -9 \\
9 -8 \\
9 -7 \\
9 -6 \\
9 -5 \\
9 -4 \\
9 -3 \\
9 -2 \\
9 -1 \\
9 0 \\
9 1 \\
9 2 \\
9 3 \\
9 4 \\
9 5 \\
9 6 \\
9 7 \\
9 8 \\
9 9 \\
9 10 \\
9 11 \\
10 -11 \\
10 -10 \\
10 -9 \\
10 -8 \\
10 -7 \\
10 -6 \\
10 -5 \\
10 -4 \\
10 -3 \\
10 -2 \\
10 -1 \\
10 0 \\
10 1 \\
10 2 \\
10 3 \\
10 4 \\
10 5 \\
10 6 \\
10 7 \\
10 8 \\
10 9 \\
10 10 \\
10 11 \\
11 -11 \\
11 -10 \\
11 -9 \\
11 -8 \\
11 -7 \\
11 -6 \\
11 -5 \\
11 -4 \\
11 -3 \\
11 -2 \\
11 -1 \\
11 0 \\
11 1 \\
11 2 \\
11 3 \\
11 4 \\
11 5 \\
11 6 \\
11 7 \\
11 8 \\
11 9 \\
11 10 \\
11 11 \\
};

\addplot [only marks, mark=*, mark size=1.5pt, fill=white, draw=black, line width=0.2pt]
table [row sep=\\] {
x y \\
-8 -8 \\
-8 -7 \\
-8 -6 \\
-8 -5 \\
-8 -4 \\
-8 -3 \\
-8 -2 \\
-8 -1 \\
-8 0 \\
-8 1 \\
-8 2 \\
-8 3 \\
-8 4 \\
-8 5 \\
-8 6 \\
-8 7 \\
-8 8 \\
-8 9 \\
-8 10 \\
-8 11 \\
-7 -8 \\
-7 -7 \\
-7 -6 \\
-7 -5 \\
-7 -4 \\
-7 -3 \\
-7 -2 \\
-7 -1 \\
-7 0 \\
-7 1 \\
-7 2 \\
-7 3 \\
-7 4 \\
-7 5 \\
-7 6 \\
-7 7 \\
-7 8 \\
-7 9 \\
-7 10 \\
-7 11 \\
-6 -8 \\
-6 -7 \\
-6 -6 \\
-6 -5 \\
-6 -4 \\
-6 -3 \\
-6 -2 \\
-6 -1 \\
-6 0 \\
-6 1 \\
-6 2 \\
-6 3 \\
-6 4 \\
-6 5 \\
-6 6 \\
-6 7 \\
-6 8 \\
-6 9 \\
-6 10 \\
-6 11 \\
-5 -8 \\
-5 -7 \\
-5 -6 \\
-5 -5 \\
-5 -4 \\
-5 -3 \\
-5 -2 \\
-5 -1 \\
-5 0 \\
-5 1 \\
-5 2 \\
-5 3 \\
-5 4 \\
-5 5 \\
-5 6 \\
-5 7 \\
-5 8 \\
-5 9 \\
-5 10 \\
-5 11 \\
-4 -8 \\
-4 -7 \\
-4 -6 \\
-4 -5 \\
-4 -4 \\
-4 -3 \\
-4 -2 \\
-4 -1 \\
-4 0 \\
-4 1 \\
-4 2 \\
-4 3 \\
-4 4 \\
-4 5 \\
-4 6 \\
-4 7 \\
-4 8 \\
-4 9 \\
-4 10 \\
-4 11 \\
-3 -8 \\
-3 -7 \\
-3 -6 \\
-3 -5 \\
-3 -4 \\
-3 -3 \\
-3 -2 \\
-3 -1 \\
-3 0 \\
-3 1 \\
-3 2 \\
-3 3 \\
-3 4 \\
-3 5 \\
-3 6 \\
-3 7 \\
-3 8 \\
-3 9 \\
-3 10 \\
-3 11 \\
-2 -8 \\
-2 -7 \\
-2 -6 \\
-2 -5 \\
-2 -4 \\
-2 -3 \\
-2 -2 \\
-2 -1 \\
-2 0 \\
-2 1 \\
-2 2 \\
-2 3 \\
-2 4 \\
-2 5 \\
-2 6 \\
-2 7 \\
-2 8 \\
-2 9 \\
-2 10 \\
-2 11 \\
-1 -8 \\
-1 -7 \\
-1 -6 \\
-1 -5 \\
-1 -4 \\
-1 -3 \\
-1 -2 \\
-1 -1 \\
-1 0 \\
-1 1 \\
-1 2 \\
-1 3 \\
-1 4 \\
-1 5 \\
-1 6 \\
-1 7 \\
-1 8 \\
-1 9 \\
-1 10 \\
-1 11 \\
0 -8 \\
0 -7 \\
0 -6 \\
0 -5 \\
0 -4 \\
0 -3 \\
0 -2 \\
0 -1 \\
0 0 \\
0 1 \\
0 2 \\
0 3 \\
0 4 \\
0 5 \\
0 6 \\
0 7 \\
0 8 \\
0 9 \\
0 10 \\
0 11 \\
1 -8 \\
1 -7 \\
1 -6 \\
1 -5 \\
1 -4 \\
1 -3 \\
1 -2 \\
1 -1 \\
1 0 \\
1 1 \\
1 2 \\
1 3 \\
1 4 \\
1 5 \\
1 6 \\
1 7 \\
1 8 \\
1 9 \\
1 10 \\
1 11 \\
2 -8 \\
2 -7 \\
2 -6 \\
2 -5 \\
2 -4 \\
2 -3 \\
2 -2 \\
2 -1 \\
2 0 \\
2 1 \\
2 2 \\
2 3 \\
2 4 \\
2 5 \\
2 6 \\
2 7 \\
2 8 \\
2 9 \\
2 10 \\
2 11 \\
3 -8 \\
3 -7 \\
3 -6 \\
3 -5 \\
3 -4 \\
3 -3 \\
3 -2 \\
3 -1 \\
3 0 \\
3 1 \\
3 2 \\
3 3 \\
3 4 \\
3 5 \\
3 6 \\
3 7 \\
3 8 \\
3 9 \\
3 10 \\
3 11 \\
4 -8 \\
4 -7 \\
4 -6 \\
4 -5 \\
4 -4 \\
4 -3 \\
4 -2 \\
4 -1 \\
4 0 \\
4 1 \\
4 2 \\
4 3 \\
4 4 \\
4 5 \\
4 6 \\
4 7 \\
4 8 \\
4 9 \\
4 10 \\
4 11 \\
5 -8 \\
5 -7 \\
5 -6 \\
5 -5 \\
5 -4 \\
5 -3 \\
5 -2 \\
5 -1 \\
5 0 \\
5 1 \\
5 2 \\
5 3 \\
5 4 \\
5 5 \\
5 6 \\
5 7 \\
5 8 \\
5 9 \\
5 10 \\
5 11 \\
6 -8 \\
6 -7 \\
6 -6 \\
6 -5 \\
6 -4 \\
6 -3 \\
6 -2 \\
6 -1 \\
6 0 \\
6 1 \\
6 2 \\
6 3 \\
6 4 \\
6 5 \\
6 6 \\
6 7 \\
6 8 \\
6 9 \\
6 10 \\
6 11 \\
7 -8 \\
7 -7 \\
7 -6 \\
7 -5 \\
7 -4 \\
7 -3 \\
7 -2 \\
7 -1 \\
7 0 \\
7 1 \\
7 2 \\
7 3 \\
7 4 \\
7 5 \\
7 6 \\
7 7 \\
7 8 \\
7 9 \\
7 10 \\
7 11 \\
8 -8 \\
8 -7 \\
8 -6 \\
8 -5 \\
8 -4 \\
8 -3 \\
8 -2 \\
8 -1 \\
8 0 \\
8 1 \\
8 2 \\
8 3 \\
8 4 \\
8 5 \\
8 6 \\
8 7 \\
8 8 \\
8 9 \\
8 10 \\
8 11 \\
9 -8 \\
9 -7 \\
9 -6 \\
9 -5 \\
9 -4 \\
9 -3 \\
9 -2 \\
9 -1 \\
9 0 \\
9 1 \\
9 2 \\
9 3 \\
9 4 \\
9 5 \\
9 6 \\
9 7 \\
9 8 \\
9 9 \\
9 10 \\
9 11 \\
10 -8 \\
10 -7 \\
10 -6 \\
10 -5 \\
10 -4 \\
10 -3 \\
10 -2 \\
10 -1 \\
10 0 \\
10 1 \\
10 2 \\
10 3 \\
10 4 \\
10 5 \\
10 6 \\
10 7 \\
10 8 \\
10 9 \\
10 10 \\
10 11 \\
11 -8 \\
11 -7 \\
11 -6 \\
11 -5 \\
11 -4 \\
11 -3 \\
11 -2 \\
11 -1 \\
11 0 \\
11 1 \\
11 2 \\
11 3 \\
11 4 \\
11 5 \\
11 6 \\
11 7 \\
11 8 \\
11 9 \\
11 10 \\
11 11 \\
};

\addplot [only marks, mark=*, mark size=1.5pt, fill=blue!40, draw=black, line width=0.2pt]
table [row sep=\\] {
x y \\
-2 -2 \\
-2 -1 \\
-2 0 \\
-2 1 \\
-2 2 \\
-1 -2 \\
-1 -1 \\
-1 0 \\
-1 1 \\
-1 2 \\
0 -2 \\
0 -1 \\
0 0 \\
0 1 \\
0 2 \\
1 -2 \\
1 -1 \\
1 0 \\
1 1 \\
1 2 \\
2 -2 \\
2 -1 \\
2 0 \\
2 1 \\
2 2 \\
};
\end{axis}
\end{tikzpicture}

\begin{tikzpicture}
\begin{axis}[
    width=4cm,
    height=4cm,
    xmin=-11, xmax=11,
    ymin=-11, ymax=11,
    axis equal,
    title={$j = (-7,-5)$, $\gamma = 1$},
    title style={font=\small},
    xlabel={}, ylabel={},
    xtick=\empty, ytick=\empty,
]

\addplot [only marks, mark=*, mark size=1.5pt, fill=gray!50, draw=black, line width=0.2pt]
table [row sep=\\] {
x y \\
-11 -11 \\
-11 -10 \\
-11 -9 \\
-11 -8 \\
-11 -7 \\
-11 -6 \\
-11 -5 \\
-11 -4 \\
-11 -3 \\
-11 -2 \\
-11 -1 \\
-11 0 \\
-11 1 \\
-11 2 \\
-11 3 \\
-11 4 \\
-11 5 \\
-11 6 \\
-11 7 \\
-11 8 \\
-11 9 \\
-11 10 \\
-11 11 \\
-10 -11 \\
-10 -10 \\
-10 -9 \\
-10 -8 \\
-10 -7 \\
-10 -6 \\
-10 -5 \\
-10 -4 \\
-10 -3 \\
-10 -2 \\
-10 -1 \\
-10 0 \\
-10 1 \\
-10 2 \\
-10 3 \\
-10 4 \\
-10 5 \\
-10 6 \\
-10 7 \\
-10 8 \\
-10 9 \\
-10 10 \\
-10 11 \\
-9 -11 \\
-9 -10 \\
-9 -9 \\
-9 -8 \\
-9 -7 \\
-9 -6 \\
-9 -5 \\
-9 -4 \\
-9 -3 \\
-9 -2 \\
-9 -1 \\
-9 0 \\
-9 1 \\
-9 2 \\
-9 3 \\
-9 4 \\
-9 5 \\
-9 6 \\
-9 7 \\
-9 8 \\
-9 9 \\
-9 10 \\
-9 11 \\
-8 -11 \\
-8 -10 \\
-8 -9 \\
-8 -8 \\
-8 -7 \\
-8 -6 \\
-8 -5 \\
-8 -4 \\
-8 -3 \\
-8 -2 \\
-8 -1 \\
-8 0 \\
-8 1 \\
-8 2 \\
-8 3 \\
-8 4 \\
-8 5 \\
-8 6 \\
-8 7 \\
-8 8 \\
-8 9 \\
-8 10 \\
-8 11 \\
-7 -11 \\
-7 -10 \\
-7 -9 \\
-7 -8 \\
-7 -7 \\
-7 -6 \\
-7 -5 \\
-7 -4 \\
-7 -3 \\
-7 -2 \\
-7 -1 \\
-7 0 \\
-7 1 \\
-7 2 \\
-7 3 \\
-7 4 \\
-7 5 \\
-7 6 \\
-7 7 \\
-7 8 \\
-7 9 \\
-7 10 \\
-7 11 \\
-6 -11 \\
-6 -10 \\
-6 -9 \\
-6 -8 \\
-6 -7 \\
-6 -6 \\
-6 -5 \\
-6 -4 \\
-6 -3 \\
-6 -2 \\
-6 -1 \\
-6 0 \\
-6 1 \\
-6 2 \\
-6 3 \\
-6 4 \\
-6 5 \\
-6 6 \\
-6 7 \\
-6 8 \\
-6 9 \\
-6 10 \\
-6 11 \\
-5 -11 \\
-5 -10 \\
-5 -9 \\
-5 -8 \\
-5 -7 \\
-5 -6 \\
-5 -5 \\
-5 -4 \\
-5 -3 \\
-5 -2 \\
-5 -1 \\
-5 0 \\
-5 1 \\
-5 2 \\
-5 3 \\
-5 4 \\
-5 5 \\
-5 6 \\
-5 7 \\
-5 8 \\
-5 9 \\
-5 10 \\
-5 11 \\
-4 -11 \\
-4 -10 \\
-4 -9 \\
-4 -8 \\
-4 -7 \\
-4 -6 \\
-4 -5 \\
-4 -4 \\
-4 -3 \\
-4 -2 \\
-4 -1 \\
-4 0 \\
-4 1 \\
-4 2 \\
-4 3 \\
-4 4 \\
-4 5 \\
-4 6 \\
-4 7 \\
-4 8 \\
-4 9 \\
-4 10 \\
-4 11 \\
-3 -11 \\
-3 -10 \\
-3 -9 \\
-3 -8 \\
-3 -7 \\
-3 -6 \\
-3 -5 \\
-3 -4 \\
-3 -3 \\
-3 -2 \\
-3 -1 \\
-3 0 \\
-3 1 \\
-3 2 \\
-3 3 \\
-3 4 \\
-3 5 \\
-3 6 \\
-3 7 \\
-3 8 \\
-3 9 \\
-3 10 \\
-3 11 \\
-2 -11 \\
-2 -10 \\
-2 -9 \\
-2 -8 \\
-2 -7 \\
-2 -6 \\
-2 -5 \\
-2 -4 \\
-2 -3 \\
-2 -2 \\
-2 -1 \\
-2 0 \\
-2 1 \\
-2 2 \\
-2 3 \\
-2 4 \\
-2 5 \\
-2 6 \\
-2 7 \\
-2 8 \\
-2 9 \\
-2 10 \\
-2 11 \\
-1 -11 \\
-1 -10 \\
-1 -9 \\
-1 -8 \\
-1 -7 \\
-1 -6 \\
-1 -5 \\
-1 -4 \\
-1 -3 \\
-1 -2 \\
-1 -1 \\
-1 0 \\
-1 1 \\
-1 2 \\
-1 3 \\
-1 4 \\
-1 5 \\
-1 6 \\
-1 7 \\
-1 8 \\
-1 9 \\
-1 10 \\
-1 11 \\
0 -11 \\
0 -10 \\
0 -9 \\
0 -8 \\
0 -7 \\
0 -6 \\
0 -5 \\
0 -4 \\
0 -3 \\
0 -2 \\
0 -1 \\
0 0 \\
0 1 \\
0 2 \\
0 3 \\
0 4 \\
0 5 \\
0 6 \\
0 7 \\
0 8 \\
0 9 \\
0 10 \\
0 11 \\
1 -11 \\
1 -10 \\
1 -9 \\
1 -8 \\
1 -7 \\
1 -6 \\
1 -5 \\
1 -4 \\
1 -3 \\
1 -2 \\
1 -1 \\
1 0 \\
1 1 \\
1 2 \\
1 3 \\
1 4 \\
1 5 \\
1 6 \\
1 7 \\
1 8 \\
1 9 \\
1 10 \\
1 11 \\
2 -11 \\
2 -10 \\
2 -9 \\
2 -8 \\
2 -7 \\
2 -6 \\
2 -5 \\
2 -4 \\
2 -3 \\
2 -2 \\
2 -1 \\
2 0 \\
2 1 \\
2 2 \\
2 3 \\
2 4 \\
2 5 \\
2 6 \\
2 7 \\
2 8 \\
2 9 \\
2 10 \\
2 11 \\
3 -11 \\
3 -10 \\
3 -9 \\
3 -8 \\
3 -7 \\
3 -6 \\
3 -5 \\
3 -4 \\
3 -3 \\
3 -2 \\
3 -1 \\
3 0 \\
3 1 \\
3 2 \\
3 3 \\
3 4 \\
3 5 \\
3 6 \\
3 7 \\
3 8 \\
3 9 \\
3 10 \\
3 11 \\
4 -11 \\
4 -10 \\
4 -9 \\
4 -8 \\
4 -7 \\
4 -6 \\
4 -5 \\
4 -4 \\
4 -3 \\
4 -2 \\
4 -1 \\
4 0 \\
4 1 \\
4 2 \\
4 3 \\
4 4 \\
4 5 \\
4 6 \\
4 7 \\
4 8 \\
4 9 \\
4 10 \\
4 11 \\
5 -11 \\
5 -10 \\
5 -9 \\
5 -8 \\
5 -7 \\
5 -6 \\
5 -5 \\
5 -4 \\
5 -3 \\
5 -2 \\
5 -1 \\
5 0 \\
5 1 \\
5 2 \\
5 3 \\
5 4 \\
5 5 \\
5 6 \\
5 7 \\
5 8 \\
5 9 \\
5 10 \\
5 11 \\
6 -11 \\
6 -10 \\
6 -9 \\
6 -8 \\
6 -7 \\
6 -6 \\
6 -5 \\
6 -4 \\
6 -3 \\
6 -2 \\
6 -1 \\
6 0 \\
6 1 \\
6 2 \\
6 3 \\
6 4 \\
6 5 \\
6 6 \\
6 7 \\
6 8 \\
6 9 \\
6 10 \\
6 11 \\
7 -11 \\
7 -10 \\
7 -9 \\
7 -8 \\
7 -7 \\
7 -6 \\
7 -5 \\
7 -4 \\
7 -3 \\
7 -2 \\
7 -1 \\
7 0 \\
7 1 \\
7 2 \\
7 3 \\
7 4 \\
7 5 \\
7 6 \\
7 7 \\
7 8 \\
7 9 \\
7 10 \\
7 11 \\
8 -11 \\
8 -10 \\
8 -9 \\
8 -8 \\
8 -7 \\
8 -6 \\
8 -5 \\
8 -4 \\
8 -3 \\
8 -2 \\
8 -1 \\
8 0 \\
8 1 \\
8 2 \\
8 3 \\
8 4 \\
8 5 \\
8 6 \\
8 7 \\
8 8 \\
8 9 \\
8 10 \\
8 11 \\
9 -11 \\
9 -10 \\
9 -9 \\
9 -8 \\
9 -7 \\
9 -6 \\
9 -5 \\
9 -4 \\
9 -3 \\
9 -2 \\
9 -1 \\
9 0 \\
9 1 \\
9 2 \\
9 3 \\
9 4 \\
9 5 \\
9 6 \\
9 7 \\
9 8 \\
9 9 \\
9 10 \\
9 11 \\
10 -11 \\
10 -10 \\
10 -9 \\
10 -8 \\
10 -7 \\
10 -6 \\
10 -5 \\
10 -4 \\
10 -3 \\
10 -2 \\
10 -1 \\
10 0 \\
10 1 \\
10 2 \\
10 3 \\
10 4 \\
10 5 \\
10 6 \\
10 7 \\
10 8 \\
10 9 \\
10 10 \\
10 11 \\
11 -11 \\
11 -10 \\
11 -9 \\
11 -8 \\
11 -7 \\
11 -6 \\
11 -5 \\
11 -4 \\
11 -3 \\
11 -2 \\
11 -1 \\
11 0 \\
11 1 \\
11 2 \\
11 3 \\
11 4 \\
11 5 \\
11 6 \\
11 7 \\
11 8 \\
11 9 \\
11 10 \\
11 11 \\
};

\addplot [only marks, mark=*, mark size=1.5pt, fill=white, draw=black, line width=0.2pt]
table [row sep=\\] {
x y \\
-2 -4 \\
-2 -3 \\
-2 -2 \\
-2 -1 \\
-2 0 \\
-2 1 \\
-2 2 \\
-2 3 \\
-2 4 \\
-2 5 \\
-2 6 \\
-2 7 \\
-2 8 \\
-2 9 \\
-2 10 \\
-2 11 \\
-2 12 \\
-2 13 \\
-2 14 \\
-2 15 \\
-1 -4 \\
-1 -3 \\
-1 -2 \\
-1 -1 \\
-1 0 \\
-1 1 \\
-1 2 \\
-1 3 \\
-1 4 \\
-1 5 \\
-1 6 \\
-1 7 \\
-1 8 \\
-1 9 \\
-1 10 \\
-1 11 \\
-1 12 \\
-1 13 \\
-1 14 \\
-1 15 \\
0 -4 \\
0 -3 \\
0 -2 \\
0 -1 \\
0 0 \\
0 1 \\
0 2 \\
0 3 \\
0 4 \\
0 5 \\
0 6 \\
0 7 \\
0 8 \\
0 9 \\
0 10 \\
0 11 \\
0 12 \\
0 13 \\
0 14 \\
0 15 \\
1 -4 \\
1 -3 \\
1 -2 \\
1 -1 \\
1 0 \\
1 1 \\
1 2 \\
1 3 \\
1 4 \\
1 5 \\
1 6 \\
1 7 \\
1 8 \\
1 9 \\
1 10 \\
1 11 \\
1 12 \\
1 13 \\
1 14 \\
1 15 \\
2 -4 \\
2 -3 \\
2 -2 \\
2 -1 \\
2 0 \\
2 1 \\
2 2 \\
2 3 \\
2 4 \\
2 5 \\
2 6 \\
2 7 \\
2 8 \\
2 9 \\
2 10 \\
2 11 \\
2 12 \\
2 13 \\
2 14 \\
2 15 \\
3 -4 \\
3 -3 \\
3 -2 \\
3 -1 \\
3 0 \\
3 1 \\
3 2 \\
3 3 \\
3 4 \\
3 5 \\
3 6 \\
3 7 \\
3 8 \\
3 9 \\
3 10 \\
3 11 \\
3 12 \\
3 13 \\
3 14 \\
3 15 \\
4 -4 \\
4 -3 \\
4 -2 \\
4 -1 \\
4 0 \\
4 1 \\
4 2 \\
4 3 \\
4 4 \\
4 5 \\
4 6 \\
4 7 \\
4 8 \\
4 9 \\
4 10 \\
4 11 \\
4 12 \\
4 13 \\
4 14 \\
4 15 \\
5 -4 \\
5 -3 \\
5 -2 \\
5 -1 \\
5 0 \\
5 1 \\
5 2 \\
5 3 \\
5 4 \\
5 5 \\
5 6 \\
5 7 \\
5 8 \\
5 9 \\
5 10 \\
5 11 \\
5 12 \\
5 13 \\
5 14 \\
5 15 \\
6 -4 \\
6 -3 \\
6 -2 \\
6 -1 \\
6 0 \\
6 1 \\
6 2 \\
6 3 \\
6 4 \\
6 5 \\
6 6 \\
6 7 \\
6 8 \\
6 9 \\
6 10 \\
6 11 \\
6 12 \\
6 13 \\
6 14 \\
6 15 \\
7 -4 \\
7 -3 \\
7 -2 \\
7 -1 \\
7 0 \\
7 1 \\
7 2 \\
7 3 \\
7 4 \\
7 5 \\
7 6 \\
7 7 \\
7 8 \\
7 9 \\
7 10 \\
7 11 \\
7 12 \\
7 13 \\
7 14 \\
7 15 \\
8 -4 \\
8 -3 \\
8 -2 \\
8 -1 \\
8 0 \\
8 1 \\
8 2 \\
8 3 \\
8 4 \\
8 5 \\
8 6 \\
8 7 \\
8 8 \\
8 9 \\
8 10 \\
8 11 \\
8 12 \\
8 13 \\
8 14 \\
8 15 \\
9 -4 \\
9 -3 \\
9 -2 \\
9 -1 \\
9 0 \\
9 1 \\
9 2 \\
9 3 \\
9 4 \\
9 5 \\
9 6 \\
9 7 \\
9 8 \\
9 9 \\
9 10 \\
9 11 \\
9 12 \\
9 13 \\
9 14 \\
9 15 \\
10 -4 \\
10 -3 \\
10 -2 \\
10 -1 \\
10 0 \\
10 1 \\
10 2 \\
10 3 \\
10 4 \\
10 5 \\
10 6 \\
10 7 \\
10 8 \\
10 9 \\
10 10 \\
10 11 \\
10 12 \\
10 13 \\
10 14 \\
10 15 \\
11 -4 \\
11 -3 \\
11 -2 \\
11 -1 \\
11 0 \\
11 1 \\
11 2 \\
11 3 \\
11 4 \\
11 5 \\
11 6 \\
11 7 \\
11 8 \\
11 9 \\
11 10 \\
11 11 \\
11 12 \\
11 13 \\
11 14 \\
11 15 \\
12 -4 \\
12 -3 \\
12 -2 \\
12 -1 \\
12 0 \\
12 1 \\
12 2 \\
12 3 \\
12 4 \\
12 5 \\
12 6 \\
12 7 \\
12 8 \\
12 9 \\
12 10 \\
12 11 \\
12 12 \\
12 13 \\
12 14 \\
12 15 \\
13 -4 \\
13 -3 \\
13 -2 \\
13 -1 \\
13 0 \\
13 1 \\
13 2 \\
13 3 \\
13 4 \\
13 5 \\
13 6 \\
13 7 \\
13 8 \\
13 9 \\
13 10 \\
13 11 \\
13 12 \\
13 13 \\
13 14 \\
13 15 \\
14 -4 \\
14 -3 \\
14 -2 \\
14 -1 \\
14 0 \\
14 1 \\
14 2 \\
14 3 \\
14 4 \\
14 5 \\
14 6 \\
14 7 \\
14 8 \\
14 9 \\
14 10 \\
14 11 \\
14 12 \\
14 13 \\
14 14 \\
14 15 \\
15 -4 \\
15 -3 \\
15 -2 \\
15 -1 \\
15 0 \\
15 1 \\
15 2 \\
15 3 \\
15 4 \\
15 5 \\
15 6 \\
15 7 \\
15 8 \\
15 9 \\
15 10 \\
15 11 \\
15 12 \\
15 13 \\
15 14 \\
15 15 \\
16 -4 \\
16 -3 \\
16 -2 \\
16 -1 \\
16 0 \\
16 1 \\
16 2 \\
16 3 \\
16 4 \\
16 5 \\
16 6 \\
16 7 \\
16 8 \\
16 9 \\
16 10 \\
16 11 \\
16 12 \\
16 13 \\
16 14 \\
16 15 \\
17 -4 \\
17 -3 \\
17 -2 \\
17 -1 \\
17 0 \\
17 1 \\
17 2 \\
17 3 \\
17 4 \\
17 5 \\
17 6 \\
17 7 \\
17 8 \\
17 9 \\
17 10 \\
17 11 \\
17 12 \\
17 13 \\
17 14 \\
17 15 \\
};

\addplot [only marks, mark=*, mark size=1.5pt, fill=blue!40, draw=black, line width=0.2pt]
table [row sep=\\] {
x y \\
-2 -2 \\
-2 -1 \\
-2 0 \\
-2 1 \\
-2 2 \\
-1 -2 \\
-1 -1 \\
-1 0 \\
-1 1 \\
-1 2 \\
0 -2 \\
0 -1 \\
0 0 \\
0 1 \\
0 2 \\
1 -2 \\
1 -1 \\
1 0 \\
1 1 \\
1 2 \\
2 -2 \\
2 -1 \\
2 0 \\
2 1 \\
2 2 \\
};
\end{axis}
\end{tikzpicture}
\begin{tikzpicture}
\begin{axis}[
    width=4cm,
    height=4cm,
    xmin=-11, xmax=11,
    ymin=-11, ymax=11,
    axis equal,
    title={$j = (9,-8)$, $\gamma = 0$},
    title style={font=\small},
    xlabel={}, ylabel={},
    xtick=\empty, ytick=\empty,
]

\addplot [only marks, mark=*, mark size=1.5pt, fill=gray!50, draw=black, line width=0.2pt]
table [row sep=\\] {
x y \\
-11 -11 \\
-11 -10 \\
-11 -9 \\
-11 -8 \\
-11 -7 \\
-11 -6 \\
-11 -5 \\
-11 -4 \\
-11 -3 \\
-11 -2 \\
-11 -1 \\
-11 0 \\
-11 1 \\
-11 2 \\
-11 3 \\
-11 4 \\
-11 5 \\
-11 6 \\
-11 7 \\
-11 8 \\
-11 9 \\
-11 10 \\
-11 11 \\
-10 -11 \\
-10 -10 \\
-10 -9 \\
-10 -8 \\
-10 -7 \\
-10 -6 \\
-10 -5 \\
-10 -4 \\
-10 -3 \\
-10 -2 \\
-10 -1 \\
-10 0 \\
-10 1 \\
-10 2 \\
-10 3 \\
-10 4 \\
-10 5 \\
-10 6 \\
-10 7 \\
-10 8 \\
-10 9 \\
-10 10 \\
-10 11 \\
-9 -11 \\
-9 -10 \\
-9 -9 \\
-9 -8 \\
-9 -7 \\
-9 -6 \\
-9 -5 \\
-9 -4 \\
-9 -3 \\
-9 -2 \\
-9 -1 \\
-9 0 \\
-9 1 \\
-9 2 \\
-9 3 \\
-9 4 \\
-9 5 \\
-9 6 \\
-9 7 \\
-9 8 \\
-9 9 \\
-9 10 \\
-9 11 \\
-8 -11 \\
-8 -10 \\
-8 -9 \\
-8 -8 \\
-8 -7 \\
-8 -6 \\
-8 -5 \\
-8 -4 \\
-8 -3 \\
-8 -2 \\
-8 -1 \\
-8 0 \\
-8 1 \\
-8 2 \\
-8 3 \\
-8 4 \\
-8 5 \\
-8 6 \\
-8 7 \\
-8 8 \\
-8 9 \\
-8 10 \\
-8 11 \\
-7 -11 \\
-7 -10 \\
-7 -9 \\
-7 -8 \\
-7 -7 \\
-7 -6 \\
-7 -5 \\
-7 -4 \\
-7 -3 \\
-7 -2 \\
-7 -1 \\
-7 0 \\
-7 1 \\
-7 2 \\
-7 3 \\
-7 4 \\
-7 5 \\
-7 6 \\
-7 7 \\
-7 8 \\
-7 9 \\
-7 10 \\
-7 11 \\
-6 -11 \\
-6 -10 \\
-6 -9 \\
-6 -8 \\
-6 -7 \\
-6 -6 \\
-6 -5 \\
-6 -4 \\
-6 -3 \\
-6 -2 \\
-6 -1 \\
-6 0 \\
-6 1 \\
-6 2 \\
-6 3 \\
-6 4 \\
-6 5 \\
-6 6 \\
-6 7 \\
-6 8 \\
-6 9 \\
-6 10 \\
-6 11 \\
-5 -11 \\
-5 -10 \\
-5 -9 \\
-5 -8 \\
-5 -7 \\
-5 -6 \\
-5 -5 \\
-5 -4 \\
-5 -3 \\
-5 -2 \\
-5 -1 \\
-5 0 \\
-5 1 \\
-5 2 \\
-5 3 \\
-5 4 \\
-5 5 \\
-5 6 \\
-5 7 \\
-5 8 \\
-5 9 \\
-5 10 \\
-5 11 \\
-4 -11 \\
-4 -10 \\
-4 -9 \\
-4 -8 \\
-4 -7 \\
-4 -6 \\
-4 -5 \\
-4 -4 \\
-4 -3 \\
-4 -2 \\
-4 -1 \\
-4 0 \\
-4 1 \\
-4 2 \\
-4 3 \\
-4 4 \\
-4 5 \\
-4 6 \\
-4 7 \\
-4 8 \\
-4 9 \\
-4 10 \\
-4 11 \\
-3 -11 \\
-3 -10 \\
-3 -9 \\
-3 -8 \\
-3 -7 \\
-3 -6 \\
-3 -5 \\
-3 -4 \\
-3 -3 \\
-3 -2 \\
-3 -1 \\
-3 0 \\
-3 1 \\
-3 2 \\
-3 3 \\
-3 4 \\
-3 5 \\
-3 6 \\
-3 7 \\
-3 8 \\
-3 9 \\
-3 10 \\
-3 11 \\
-2 -11 \\
-2 -10 \\
-2 -9 \\
-2 -8 \\
-2 -7 \\
-2 -6 \\
-2 -5 \\
-2 -4 \\
-2 -3 \\
-2 -2 \\
-2 -1 \\
-2 0 \\
-2 1 \\
-2 2 \\
-2 3 \\
-2 4 \\
-2 5 \\
-2 6 \\
-2 7 \\
-2 8 \\
-2 9 \\
-2 10 \\
-2 11 \\
-1 -11 \\
-1 -10 \\
-1 -9 \\
-1 -8 \\
-1 -7 \\
-1 -6 \\
-1 -5 \\
-1 -4 \\
-1 -3 \\
-1 -2 \\
-1 -1 \\
-1 0 \\
-1 1 \\
-1 2 \\
-1 3 \\
-1 4 \\
-1 5 \\
-1 6 \\
-1 7 \\
-1 8 \\
-1 9 \\
-1 10 \\
-1 11 \\
0 -11 \\
0 -10 \\
0 -9 \\
0 -8 \\
0 -7 \\
0 -6 \\
0 -5 \\
0 -4 \\
0 -3 \\
0 -2 \\
0 -1 \\
0 0 \\
0 1 \\
0 2 \\
0 3 \\
0 4 \\
0 5 \\
0 6 \\
0 7 \\
0 8 \\
0 9 \\
0 10 \\
0 11 \\
1 -11 \\
1 -10 \\
1 -9 \\
1 -8 \\
1 -7 \\
1 -6 \\
1 -5 \\
1 -4 \\
1 -3 \\
1 -2 \\
1 -1 \\
1 0 \\
1 1 \\
1 2 \\
1 3 \\
1 4 \\
1 5 \\
1 6 \\
1 7 \\
1 8 \\
1 9 \\
1 10 \\
1 11 \\
2 -11 \\
2 -10 \\
2 -9 \\
2 -8 \\
2 -7 \\
2 -6 \\
2 -5 \\
2 -4 \\
2 -3 \\
2 -2 \\
2 -1 \\
2 0 \\
2 1 \\
2 2 \\
2 3 \\
2 4 \\
2 5 \\
2 6 \\
2 7 \\
2 8 \\
2 9 \\
2 10 \\
2 11 \\
3 -11 \\
3 -10 \\
3 -9 \\
3 -8 \\
3 -7 \\
3 -6 \\
3 -5 \\
3 -4 \\
3 -3 \\
3 -2 \\
3 -1 \\
3 0 \\
3 1 \\
3 2 \\
3 3 \\
3 4 \\
3 5 \\
3 6 \\
3 7 \\
3 8 \\
3 9 \\
3 10 \\
3 11 \\
4 -11 \\
4 -10 \\
4 -9 \\
4 -8 \\
4 -7 \\
4 -6 \\
4 -5 \\
4 -4 \\
4 -3 \\
4 -2 \\
4 -1 \\
4 0 \\
4 1 \\
4 2 \\
4 3 \\
4 4 \\
4 5 \\
4 6 \\
4 7 \\
4 8 \\
4 9 \\
4 10 \\
4 11 \\
5 -11 \\
5 -10 \\
5 -9 \\
5 -8 \\
5 -7 \\
5 -6 \\
5 -5 \\
5 -4 \\
5 -3 \\
5 -2 \\
5 -1 \\
5 0 \\
5 1 \\
5 2 \\
5 3 \\
5 4 \\
5 5 \\
5 6 \\
5 7 \\
5 8 \\
5 9 \\
5 10 \\
5 11 \\
6 -11 \\
6 -10 \\
6 -9 \\
6 -8 \\
6 -7 \\
6 -6 \\
6 -5 \\
6 -4 \\
6 -3 \\
6 -2 \\
6 -1 \\
6 0 \\
6 1 \\
6 2 \\
6 3 \\
6 4 \\
6 5 \\
6 6 \\
6 7 \\
6 8 \\
6 9 \\
6 10 \\
6 11 \\
7 -11 \\
7 -10 \\
7 -9 \\
7 -8 \\
7 -7 \\
7 -6 \\
7 -5 \\
7 -4 \\
7 -3 \\
7 -2 \\
7 -1 \\
7 0 \\
7 1 \\
7 2 \\
7 3 \\
7 4 \\
7 5 \\
7 6 \\
7 7 \\
7 8 \\
7 9 \\
7 10 \\
7 11 \\
8 -11 \\
8 -10 \\
8 -9 \\
8 -8 \\
8 -7 \\
8 -6 \\
8 -5 \\
8 -4 \\
8 -3 \\
8 -2 \\
8 -1 \\
8 0 \\
8 1 \\
8 2 \\
8 3 \\
8 4 \\
8 5 \\
8 6 \\
8 7 \\
8 8 \\
8 9 \\
8 10 \\
8 11 \\
9 -11 \\
9 -10 \\
9 -9 \\
9 -8 \\
9 -7 \\
9 -6 \\
9 -5 \\
9 -4 \\
9 -3 \\
9 -2 \\
9 -1 \\
9 0 \\
9 1 \\
9 2 \\
9 3 \\
9 4 \\
9 5 \\
9 6 \\
9 7 \\
9 8 \\
9 9 \\
9 10 \\
9 11 \\
10 -11 \\
10 -10 \\
10 -9 \\
10 -8 \\
10 -7 \\
10 -6 \\
10 -5 \\
10 -4 \\
10 -3 \\
10 -2 \\
10 -1 \\
10 0 \\
10 1 \\
10 2 \\
10 3 \\
10 4 \\
10 5 \\
10 6 \\
10 7 \\
10 8 \\
10 9 \\
10 10 \\
10 11 \\
11 -11 \\
11 -10 \\
11 -9 \\
11 -8 \\
11 -7 \\
11 -6 \\
11 -5 \\
11 -4 \\
11 -3 \\
11 -2 \\
11 -1 \\
11 0 \\
11 1 \\
11 2 \\
11 3 \\
11 4 \\
11 5 \\
11 6 \\
11 7 \\
11 8 \\
11 9 \\
11 10 \\
11 11 \\
};

\addplot [only marks, mark=*, mark size=1.5pt, fill=white, draw=black, line width=0.2pt]
table [row sep=\\] {
x y \\
-18 -1 \\
-18 0 \\
-18 1 \\
-18 2 \\
-18 3 \\
-18 4 \\
-18 5 \\
-18 6 \\
-18 7 \\
-18 8 \\
-18 9 \\
-18 10 \\
-18 11 \\
-18 12 \\
-18 13 \\
-18 14 \\
-18 15 \\
-18 16 \\
-18 17 \\
-18 18 \\
-17 -1 \\
-17 0 \\
-17 1 \\
-17 2 \\
-17 3 \\
-17 4 \\
-17 5 \\
-17 6 \\
-17 7 \\
-17 8 \\
-17 9 \\
-17 10 \\
-17 11 \\
-17 12 \\
-17 13 \\
-17 14 \\
-17 15 \\
-17 16 \\
-17 17 \\
-17 18 \\
-16 -1 \\
-16 0 \\
-16 1 \\
-16 2 \\
-16 3 \\
-16 4 \\
-16 5 \\
-16 6 \\
-16 7 \\
-16 8 \\
-16 9 \\
-16 10 \\
-16 11 \\
-16 12 \\
-16 13 \\
-16 14 \\
-16 15 \\
-16 16 \\
-16 17 \\
-16 18 \\
-15 -1 \\
-15 0 \\
-15 1 \\
-15 2 \\
-15 3 \\
-15 4 \\
-15 5 \\
-15 6 \\
-15 7 \\
-15 8 \\
-15 9 \\
-15 10 \\
-15 11 \\
-15 12 \\
-15 13 \\
-15 14 \\
-15 15 \\
-15 16 \\
-15 17 \\
-15 18 \\
-14 -1 \\
-14 0 \\
-14 1 \\
-14 2 \\
-14 3 \\
-14 4 \\
-14 5 \\
-14 6 \\
-14 7 \\
-14 8 \\
-14 9 \\
-14 10 \\
-14 11 \\
-14 12 \\
-14 13 \\
-14 14 \\
-14 15 \\
-14 16 \\
-14 17 \\
-14 18 \\
-13 -1 \\
-13 0 \\
-13 1 \\
-13 2 \\
-13 3 \\
-13 4 \\
-13 5 \\
-13 6 \\
-13 7 \\
-13 8 \\
-13 9 \\
-13 10 \\
-13 11 \\
-13 12 \\
-13 13 \\
-13 14 \\
-13 15 \\
-13 16 \\
-13 17 \\
-13 18 \\
-12 -1 \\
-12 0 \\
-12 1 \\
-12 2 \\
-12 3 \\
-12 4 \\
-12 5 \\
-12 6 \\
-12 7 \\
-12 8 \\
-12 9 \\
-12 10 \\
-12 11 \\
-12 12 \\
-12 13 \\
-12 14 \\
-12 15 \\
-12 16 \\
-12 17 \\
-12 18 \\
-11 -1 \\
-11 0 \\
-11 1 \\
-11 2 \\
-11 3 \\
-11 4 \\
-11 5 \\
-11 6 \\
-11 7 \\
-11 8 \\
-11 9 \\
-11 10 \\
-11 11 \\
-11 12 \\
-11 13 \\
-11 14 \\
-11 15 \\
-11 16 \\
-11 17 \\
-11 18 \\
-10 -1 \\
-10 0 \\
-10 1 \\
-10 2 \\
-10 3 \\
-10 4 \\
-10 5 \\
-10 6 \\
-10 7 \\
-10 8 \\
-10 9 \\
-10 10 \\
-10 11 \\
-10 12 \\
-10 13 \\
-10 14 \\
-10 15 \\
-10 16 \\
-10 17 \\
-10 18 \\
-9 -1 \\
-9 0 \\
-9 1 \\
-9 2 \\
-9 3 \\
-9 4 \\
-9 5 \\
-9 6 \\
-9 7 \\
-9 8 \\
-9 9 \\
-9 10 \\
-9 11 \\
-9 12 \\
-9 13 \\
-9 14 \\
-9 15 \\
-9 16 \\
-9 17 \\
-9 18 \\
-8 -1 \\
-8 0 \\
-8 1 \\
-8 2 \\
-8 3 \\
-8 4 \\
-8 5 \\
-8 6 \\
-8 7 \\
-8 8 \\
-8 9 \\
-8 10 \\
-8 11 \\
-8 12 \\
-8 13 \\
-8 14 \\
-8 15 \\
-8 16 \\
-8 17 \\
-8 18 \\
-7 -1 \\
-7 0 \\
-7 1 \\
-7 2 \\
-7 3 \\
-7 4 \\
-7 5 \\
-7 6 \\
-7 7 \\
-7 8 \\
-7 9 \\
-7 10 \\
-7 11 \\
-7 12 \\
-7 13 \\
-7 14 \\
-7 15 \\
-7 16 \\
-7 17 \\
-7 18 \\
-6 -1 \\
-6 0 \\
-6 1 \\
-6 2 \\
-6 3 \\
-6 4 \\
-6 5 \\
-6 6 \\
-6 7 \\
-6 8 \\
-6 9 \\
-6 10 \\
-6 11 \\
-6 12 \\
-6 13 \\
-6 14 \\
-6 15 \\
-6 16 \\
-6 17 \\
-6 18 \\
-5 -1 \\
-5 0 \\
-5 1 \\
-5 2 \\
-5 3 \\
-5 4 \\
-5 5 \\
-5 6 \\
-5 7 \\
-5 8 \\
-5 9 \\
-5 10 \\
-5 11 \\
-5 12 \\
-5 13 \\
-5 14 \\
-5 15 \\
-5 16 \\
-5 17 \\
-5 18 \\
-4 -1 \\
-4 0 \\
-4 1 \\
-4 2 \\
-4 3 \\
-4 4 \\
-4 5 \\
-4 6 \\
-4 7 \\
-4 8 \\
-4 9 \\
-4 10 \\
-4 11 \\
-4 12 \\
-4 13 \\
-4 14 \\
-4 15 \\
-4 16 \\
-4 17 \\
-4 18 \\
-3 -1 \\
-3 0 \\
-3 1 \\
-3 2 \\
-3 3 \\
-3 4 \\
-3 5 \\
-3 6 \\
-3 7 \\
-3 8 \\
-3 9 \\
-3 10 \\
-3 11 \\
-3 12 \\
-3 13 \\
-3 14 \\
-3 15 \\
-3 16 \\
-3 17 \\
-3 18 \\
-2 -1 \\
-2 0 \\
-2 1 \\
-2 2 \\
-2 3 \\
-2 4 \\
-2 5 \\
-2 6 \\
-2 7 \\
-2 8 \\
-2 9 \\
-2 10 \\
-2 11 \\
-2 12 \\
-2 13 \\
-2 14 \\
-2 15 \\
-2 16 \\
-2 17 \\
-2 18 \\
-1 -1 \\
-1 0 \\
-1 1 \\
-1 2 \\
-1 3 \\
-1 4 \\
-1 5 \\
-1 6 \\
-1 7 \\
-1 8 \\
-1 9 \\
-1 10 \\
-1 11 \\
-1 12 \\
-1 13 \\
-1 14 \\
-1 15 \\
-1 16 \\
-1 17 \\
-1 18 \\
0 -1 \\
0 0 \\
0 1 \\
0 2 \\
0 3 \\
0 4 \\
0 5 \\
0 6 \\
0 7 \\
0 8 \\
0 9 \\
0 10 \\
0 11 \\
0 12 \\
0 13 \\
0 14 \\
0 15 \\
0 16 \\
0 17 \\
0 18 \\
1 -1 \\
1 0 \\
1 1 \\
1 2 \\
1 3 \\
1 4 \\
1 5 \\
1 6 \\
1 7 \\
1 8 \\
1 9 \\
1 10 \\
1 11 \\
1 12 \\
1 13 \\
1 14 \\
1 15 \\
1 16 \\
1 17 \\
1 18 \\
};

\addplot [only marks, mark=*, mark size=1.5pt, fill=blue!40, draw=black, line width=0.2pt]
table [row sep=\\] {
x y \\
-2 -2 \\
-2 -1 \\
-2 0 \\
-2 1 \\
-2 2 \\
-1 -2 \\
-1 -1 \\
-1 0 \\
-1 1 \\
-1 2 \\
0 -2 \\
0 -1 \\
0 0 \\
0 1 \\
0 2 \\
1 -2 \\
1 -1 \\
1 0 \\
1 1 \\
1 2 \\
2 -2 \\
2 -1 \\
2 0 \\
2 1 \\
2 2 \\
};
\end{axis}
\end{tikzpicture}